\documentclass[envcountsect,a4paper]{llncs}

\spnewtheorem{notation}[note]{Notation}{\itshape}{}

\usepackage{amsmath}
\usepackage{amssymb}
\usepackage{stmaryrd}
\usepackage{proof}
\usepackage{graphicx}
\usepackage{xypic}
\usepackage{tikz}
\usetikzlibrary{calc}
\usepackage{enumerate}
\usepackage{mathrsfs}
\usepackage{mathtools}
\usepackage{ifthen}
\usepackage[
]{todonotes}
\usepackage{cmll}
\usepackage{wrapfig}
\usepackage{xspace}
\usepackage{adjustbox}
\adjustboxset{max width = \linewidth, center = \linewidth}
\usepackage{ebproof} 
\ebproofset{
	template={$\vdash \inserttext$},
	right label template={$\scriptstyle(\inserttext)$},
	left label template={$\scriptstyle\inserttext$},
	center=false,
}
\usepackage{color}
\usepackage{url}
\usepackage{hyperref}
\hypersetup{colorlinks,linkcolor={red},citecolor={blue},urlcolor={red}}
\usepackage[capitalize]{cleveref}

\def\eg{{\it e.g.}\xspace}
\def\ie{{\it i.e.}\xspace}
\newcommand*{\resp}{resp.}

\newcommand{\at}{\mathsf{at}}
\newcommand{\Dag}[1]{\mathtt{dag}(#1)}

\newcommand\todoc[2][]{\todo[color=pink!20,#1]{#2}} 
\newcommand\todocf[1]{\todoc[inline]{#1}}
\newcommand\todor[2][]{\todo[color=orange!20,#1]{#2}} 

\newcommand{\pink}[1]{{\color{magenta}{#1}}}

\newcommand{\RED}[1]{{\color{red}{#1}}}
\newcommand{\blue}[1]{{\color{blue}{#1}}}

\Crefname{equation}{Eq.}{Eqs.}
\Crefname{figure}{Fig.}{Figs.}
\Crefname{theorem}{Thm.}{Thm.}
\Crefname{thm}{Thm.}{Thm.}
\Crefname{proposition}{Prop.}{Prop.}
\Crefname{prop}{Prop.}{Prop.}
\Crefname{remark}{Rem.}{Rem.}
\Crefname{section}{Section}{Sections}
\Crefname{corollary}{Cor.}{Cor.}
\Crefname{cor}{Cor.}{Cor.}
\Crefname{definition}{Def.}{Def.}
\Crefname{property}{Property}{Properties}

\newcommand{\textdef}[1]{\textbf{#1}} 

\newcommand{\BigO}[1]{\mathbf{O}\left(#1\right)}
\newcommand{\Exp}[1]{\mathsf{exp}\left(#1\right)}

\newcommand{\rules}{\mathtt{R}}
\newcommand{\rr}{\mathtt{r}}
\newcommand{\Red}[1][]{\xrightarrow{#1}}	
\newcommand{\Reds}[1][]{\mathrel{\mathop{\xrightarrow{#1}}\!{}^\ast}} 
\newcommand{\FromReds}[1][]{\mathrel{\mathop{{}^\ast\xleftarrow{#1}}}} 

\newcommand{\MLL}{\ensuremath{\mathsf{MLL}}\xspace} 

\newcommand{\One}{\mathsf 1} 
\newcommand{\Bottom}{\mathsf \perp} 
\newcommand*{\orth}{^\perp} 

\newcommand{\FormA}{F} 
\newcommand{\FormB}{G} 

\newcommand{\Nm}[1]{\mathsf{Nm}(#1)}
\newcommand{\Names}{\textsf{Names}\xspace}

\newcommand{\Proj}[3][]{#2\vert_{#3}^{#1}} 

\newcommand{\Val}[1]{\mathsf{Val}(#1)} 

\newcommand{\Fprod}{\odot} 
\newcommand{\FProd}{\odot} 
\newcommand{\BigFProd}{\bigodot} 

\newcommand{\true}{\mathtt{t}}
\newcommand{\false}{\mathtt{f}}

\newcommand{\tand }{\texttt{ and }}

\newcommand{\sem}[1]{\llparenthesis {#1} \rrparenthesis}

\newcommand{\BN}{Bayesian Network\xspace}
\newcommand{\BNs}{Bayesian Networks\xspace}

\newcommand{\bnet}[1]{\mathscr{B}(#1)}

\newcommand{\netN}{\mathcal N}
\newcommand{\N}{\netN}
\newcommand{\netR}{\mathcal R}
\newcommand{\R}{\netR}
\newcommand{\netM}{\mathcal M}
\newcommand{\M}{\mathcal M}
\newcommand{\netD}{\mathcal D}
\newcommand{\D}{\netD}
\newcommand{\netG}{\mathcal G}
\newcommand{\netC}{\mathcal C}
\newcommand{\netS}{\mathcal S}

\newcommand{\node}{\mathtt{n}}

\newcommand{\tree}{\mathcal T}

\newcommand{\CutS}[1]{\mathrm{Cut}(#1)}
\newcommand{\Cuts}[1]{\mathrm{Cut}(#1)}

\newcommand{\size}[1]{\#(#1)}

\newcommand{\bX}{\mathbf X}
\newcommand{\bY}{\mathbf Y}
\newcommand{\bZ}{\mathbf Z}

\newcommand{\bW}{\mathbf W}

\newcommand{\bx}{{\mathbf x}}
\newcommand{\by}{{\mathbf y}}
\newcommand{\bz}{{\mathbf z}}

\newcommand{\pax}{box\xspace}
\newcommand{\paxs}{boxes\xspace}

\newcommand{\cut}{\mathsf{cut}}
\newcommand{\Ax}{\mathsf{Box}}

\newcommand{\Boxes}[1]{\mathsf{Boxes}(#1)}
\newcommand{\boxn}{\mathsf{box}}
\newcommand{\ax}{\mathsf{ax}}
\newcommand{\weak}{\mathsf{w}}
\newcommand{\cn}{\mathtt{c}}

\newcommand*{\defeq}{\stackrel{\text{def}}{=}}

\newcommand{\rv}{r.v.\xspace}
\newcommand{\rvs}{r.v.s\xspace}

\newcommand{\MLLB}{\ensuremath{\mathsf{MLL^{\mathtt B}}}\xspace}
\newcommand{\MLLAx}{\MLLB}

\newcommand{\At}[1]{\Nm{#1}}

\newcommand{\cutnet}{cut-net\xspace}

\newcommand{\Out}[2]{\mathsf{Split}_{#2}({#1})}

\newcommand*{\pulledout}[2]{\ensuremath{\Out{#1}{#2}}}

\newcommand{\Width}[1]{\mathrm{Width}(#1)}
\newcommand{\pol}[1]{\mathit{Pol}({#1})}

\newcommand{\Xn}{X^-}
\newcommand{\Yn}{Y^-}
\newcommand{\Zn}{Z^-}
\newcommand{\X}{X^+}
\newcommand{\Y}{Y^+}
\newcommand{\Z}{Z^+}

\newcommand{\bn}{\mathcal B}

\newcommand{\cpts}{CPT's\xspace}
\newcommand{\cpt}{CPT\xspace}

\newcommand{\Bpn}{Bayesian proof-net\xspace}
\newcommand{\Bpns}{Bayesian proof-nets\xspace}
\newcommand{\pn}{proof-net\xspace}
\newcommand{\pns}{proof-nets\xspace}
\newcommand{\BPN}{\textsf{bpn}\xspace}
\newcommand{\BPNs}{\textsf{bpn's}\xspace}
\newcommand{\phisem}[1]{\sem {#1}}
\newcommand{\pt}{proof-tree\xspace}
\newcommand{\pts}{proof-trees\xspace}

\newcommand{\bbox}{\mathtt{b}}

\newcommand{\Pa}[1]{\mathsf{Pa}(#1)}

\newcommand{\ft}[2]{  \overset{#2}{#1}  }

\newcommand{\ftone}{ \mathbf{1} }

\newcommand{\x}{ \mathtt{x}}
\newcommand{\y}{ \mathtt{y}}

\newcommand{\w}{\mathtt {w}}

\newcommand{\xs}{\bx}
\newcommand{\ys}{\by}

\tikzset{
    named edge/.style={draw=none,inner sep=1pt,outer sep=2pt,rectangle,minimum size=0pt,auto},
    named path/.style={named edge,sloped},
}

\definecolor{lasallegreen}{rgb}{0.03, 0.47, 0.19}

\newcommand{\Main}[1]{\mathsf{Main}(#1)}

\newcommand{\CPT}{CPT\xspace}

\newcommand{\Net}{\N}

\newcommand{\ProofN}{\R}

\newcommand{\arule}{\mathtt{a}}

\newcommand{\aphi}{\Phi}

\newcommand{\commenti}{0}
\newcommand{\condinc}[2]{\ifthenelse{\equal{\commenti}{0}}{#1}{\blue{#2}} }
\newcommand{\version}{0}
\newcommand{\SLV}[2]{\ifthenelse{\equal{\version}{0}}{#1}{ \RED{#2}}}

\begin{document}

\title{Bayesian Networks and Proof-Nets:\texorpdfstring{\\}{ }the proof-theory of Bayesian Inference}

\author{%
	R{\'e}mi Di Guardia%
	\and
	Thomas Ehrhard%
	\and
	J{\'e}r{\^o}me Evrard%
	\and
	Claudia Faggian%
}

\institute{IRIF, Université Paris Cité - CNRS}

\maketitle

\begin{abstract}

We study the correspondence between Bayesian Networks and graphical representation of proofs in linear logic.
The goal of this paper is threefold:
to develop a proof-theoretical account of Bayesian inference (in the spirit of the Curry-Howard correspondence between proofs and programs),
to provide compositional graphical methods,
and to take into account computational efficiency.

We exploit the fact that the decomposition of a graph is  more flexible than that of a \pt, or of a type-derivation, even if 
 compositionality becomes  more challenging. 


\end{abstract} 

\section{Introduction}



\BNs~\cite{Pearl88} are a prominent  tool for probabilistic reasoning, as they are able to  express 
large probability distributions  in a way which is  compact  (\emph{factorized representation}) and which allows for efficient inference (\emph{factorized computation}).


In this paper we propose a proof-theoretical account of \BNs and of  factorized inference, via    Linear Logic~\cite{ll,synsem}, namely its 
graph syntax---proof-nets---which embeds \emph{cut-elimination} (and its dual, cut-expansion) as a graph-rewriting process.
The benefits of an approach  bringing together \BNs and  the rich toolbox of a proof-theory which is intrinsically resources-aware,  is to   accommodate both compositional, graphical reasoning, and efficient computations. 

In  the setting of categorical probability theory, a rich body of  work pioneered by~\cite{CoeckeS12,JacobsZ16,JacobsZ} has already  disclosed  a deep connection between \BNs and logic, bringing  a  compositional  and structured approach into probabilistic reasoning (an in-depth treatment is in~\cite{JacobsBook}). A key ingredient  here is  the   graphical calculus
of string diagrams, which  provides a 
uniform  and expressive language for formal  diagrammatic reasoning,  interpreted into monoidal categories.

%
%

One aspect which is generally missing in  the categorical approach is  the space and time consumption of  probabilistic reasoning, which however is a strong  motivation behind the introduction~\cite{Pearl86} and development of \BNs. 
Recent work rooted in  Linear Logic~\cite{EhrhardFP23,EhrhardFP25,popl24} is bringing   a cost-aware perspective into semantics.   In particular, \cite{popl24} introduces a semantical framework that integrates the efficiency of Bayesian networks with the  compositional nature of type systems, by equipping  (a linear logic variant of) $\lambda$-terms  with a semantics based on \emph{factors}, the very same mathematical structures underlying  BNs and their inference algorithms.  

With the  proof-nets formalism, we aim at unlocking  the best of both approaches: efficiency  and graphical reasoning. Along the way, we gain a 
foundational understanding, where the cornerstone of proof-theory, namely \emph{cut-elimination} (and its dual, cut-expansion), has a prominent role.



\paragraph{Factorized representation and factorized computation.}

A joint distribution is a \emph{global} function involving many variables. 
A common way to deal with a  complex function is to  factorize it   as a product of \emph{local} functions, each of  which depends on a subset of the variables. In the setting of \BNs, factorization   involves both the representation and the  algorithms for (exact) inference.

%
In proof-theory, the natural way to \emph{factorize  a proof} in smaller components is to factorize  it in sub-proofs which are then \emph{composed via cut}. We follow exactly this  way.
\SLV{}{
	We then show that---remarkably---there is a tight correspondence between the \emph{decomposition of  a proof-net},
	and the well-known decomposition of Bayesian Networks into \emph{clique trees}. The inductive \emph{interpretation} of a factorized proof-net closely corresponds to \emph{message passing over clique trees} (\Cref{sec:MP}).
	The correspondence turns out to be so tight, that the respective computational costs are similar.}
In section \ref{sec:factorization},  we \emph{factorize} a proof-net in the composition of smaller nets, whose interpretation has a smaller cost (\Cref{sec:turbo}).

\paragraph{Graphical reasoning.} In \Cref{sec:conditional_indepenence} we demonstrate graphical reasoning by  providing 
an extremely simple (diagrammatic) proof of the soundness of   d-separation, a well-known graphical criterion to establish conditional independence among random variables.

\paragraph{Graph Decomposition and Graph Compositionality} Both to achieve an efficient factorization  and in graphical reasoning, we exploit the fact that the decomposition  of a graph  is much more flexible than that of a \pt (or of a type-derivation). The challenge we face is that (efficient) compositionality is less immediate (\Cref{th:compositional}),  and  associating  a \pt  to the graph decomposition is also non trivial (\Cref{sec:typing_cut_net}).

\paragraph{Cost.} In general, the cost of \emph{actually computing} the semantics explodes when taking a categorical approach, because  the   product \emph{$\otimes$}  behaves like  the tensor product of matrices (see~\cite{popl24} for examples); 
computing the semantics of $n$ binary random variables easily leads to intermediate 
computations whose size is much larger than $2^n$, the size of the \emph{full} joint distribution. The efficiency of a factors-based semantics (as in \BNs, and as in~\cite{popl24}) lies in a  definition of product (the \emph{factors product}) which \emph{is not} behaving like the tensor product of matrices.

\SLV{}{
	\paragraph{Factors-based vs   Tensor-based semantics.}\label{sec:on_product} 
	The \emph{raison d'être} of    \BNs is making probabilistic reasoning feasible by avoiding to compute a  full joint probability distribution. 
	A key ingredient are the operations on \emph{factors}, the  mathematical tools on which the computations are based. 
	
	The efficiency of a factors-based semantics (as in~\cite{popl24}) lies in a  definition of product (the \emph{factors product}) which    \emph{is not}  behaving  like the tensor product  $\otimes$ of matrices,   as in a categorical 
	setting. Of course, the product of factors can be simulated via  $\otimes$, but at a higher cost, possibly producing large  intermediate computations, as described in \Cref{ex:large} (from~\cite{popl24}). Let us first observe  the behaviour of the two products.
	\begin{example}[Not a tensor product]\label{ex:tensor}
		Consider two CPT  $\phi^{X_1} = \Pr(X_1|Y_1, Y_2,Y_3)$ 
		and $\phi^{X_2} = \Pr(X_2|Y_1, Y_2,Y_3)$, where two distinct   \rvs ($X_1$ and $X_2$, respectively) are conditioned to the \emph{same} set of \rvs  $Y_1, Y_2,Y_3$. If all variables are binary,
		each  $\phi^{X_i}$ corresponds to a matrix of size $2^4$. Computing the \textbf{tensor product}  $\phi^{X_1} \otimes \phi^{X_2}$ of the two matrices, requires to \emph{compute and store}  $2^4\cdot 2^4=2^8$ entries. In contrast,  
		the \textbf{factor product} $\phi^{X_1} \FProd \phi^{X_2}$  computes $2^5$ entries. 
	\end{example}
	The following example consider a  \BN representing a joint distribution over 5 \rvs, hence  of size $2^5$.  	 A tensor-based semantics may  perform a computation  that  is \emph{larger than  \emph{full} joint distribution}.  This is striking.
	\begin{example}[Cost of computing the semantics]\label{ex:large}
		Consider a \BN $\bn$ \RED{FIG??}  over  5 binary \rvs $X_1, X_2,Y_1, Y_2,Y_3$,  where  both $X_1$ and $X_2$ have the same parents   $Y_1, Y_2,Y_3$. Let   $ \phi^{X_1}$ and $ \phi^{X_2}$ (as in \Cref{ex:tensor} )  be the corresponding conditional probabilities.
		The categorical interpretation of $\bn$ will  (in general) pass via  $\phi^{X_1} \otimes \phi^{X_2}$, which  requires to compute and store $2^8$ values.  
	\end{example}
}

%

\subsubsection{Related work.} We  have already mentioned above   the most relevant literature related to our work.
 Here we briefly  comment on the papers which are more technically  related.
 
The syntax of \pn with boxes is introduced in~\cite{EhrhardFP23}, however while the syntax allows the encoding of \BNs, it does not characterize them (a \pn does not necessarily correspond to a \BN). We are inspired by that paper,  which  advocates cost-awareness in the computation of the  semantics. The approach in~\cite{EhrhardFP23} is to use inference algorithms to efficiently compute the denotation of any proof.
In this paper, we follow a different directions: we  focus on the proof-nets, and on inference-as-interpretation where the data structure which supports the computations is the \pn itself.

 We adopt the factor-based semantics introduced in~\cite{popl24}; while we inherit most of their results, our  notions of ``component'' and ``compositionality'' are  stronger (how we explain in \Cref{sec:sem}), yielding to  \Cref{th:compositional}.


\section{Background}
\subsection{An informal example}\label{sec:informal}

A Bayesian Network consists of two parts: a qualitative component, given by a directed acyclic graph, and a quantitative component, given by conditional probabilities.
This bears a striking resemblance with proof-nets of Linear Logic (LL)~\cite{ll,synsem}: proof-nets are a graph representation of the syntax and   cut-elimination of LL proofs, to which can naturally be associated a quantitative interpretation (see, \eg,~\cite{EhrhardT19}).

\begin{figure}
	\centering
	\begin{minipage}[b]{0.45\linewidth}
		\includegraphics[page=2,width=1\linewidth]{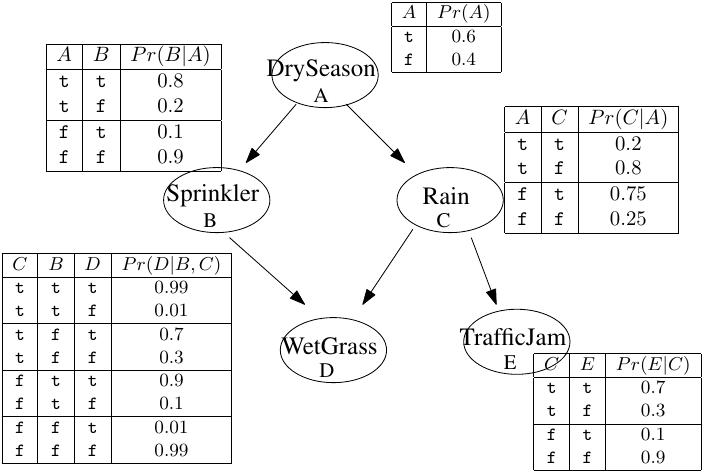}
		\caption{Example of \BN}
		\label{fig:BNrain}\label{fig:BN}
	\end{minipage}
	\begin{minipage}[b]{0.45\linewidth}
		\includegraphics[page=9,width=1\linewidth]{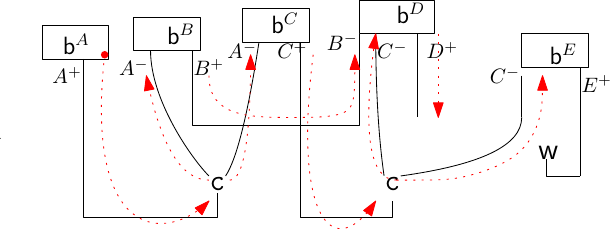}
		\caption{Proof-net $\R_D$}
		\label{fig:ABCDE_D}\label{fig:BNtoPN}
	\end{minipage}
\end{figure}

\subsubsection{An example of Bayesian Network.}


Let us start with a classical example (from~\cite{DarwicheBook}). We want
to model  the  fact that the lawn being  Wet in the morning may depend on either Rain or the Sprinkler being on. In turn, both Rain and the regulation of the  Sprinkler  depend on the Season. Moreover, Traffic Jams are correlated with Rain.
The dependencies between these  five variables (shorten  into $ A,B,C,D,E $) are pictured in  \Cref{fig:BNrain},  where 
the strength of the dependencies is quantified by 
conditional probability tables.
We wonder: did it rain last night? 
Assuming we are in DrySeason, our \emph{prior} belief is that Rain happens with probability $ 0.2 $. However, if we observe that the lawn is Wet, our confidence will increase.  The updated belief is called \emph{posterior}.
The model in \Cref{fig:BNrain} allows us to infer the \emph{posterior probability} of Rain, given the evidence, formally $\Pr(C=\true\mid D=\true)$, or to infer how likely is it that the lawn is wet, \ie\ infer the \emph{marginal} $\Pr(D=\true)$.
Conditional probabilities and marginals are typical queries which can be answered by \emph{Bayesian inference}, whose core is Bayes conditioning.

%
So, to obtain the posterior $\Pr(C=\true\mid D=\true)$, we have to compute:
\begin{itemize}
	\item the marginal $\Pr(C=\true,D=\true)$, which can be  obtained by summing out the other variables from the joint probability (\emph{marginalization});
	\item normalize by the marginal probability $ \Pr(D=\true)$ of the evidence, which is computed in a similar way.
\end{itemize}



\condinc{}{
	\[	\begin{array}{|c|c|c|}
		\hline
		D&	R	&Pr(R|D)  \\
		\hline
		\true&	\true & 0.2 \\
		\true&	\false	& 0.8\\
		\hline
		\false	&	\true & 0.75 \\
		\false	&	\false	& 0.25\\
		\hline	
	\end{array}
	\quad
	\begin{array}{|c|c|c|}
		\hline
		D&	R	&Pr(R|D)  \\
		\hline
		\true&	\true & 0.8 \\
		\true&	\false	& 0.2\\
		\hline
		\false	&	\true & 0.1 \\
		\false	&	\false	& 0.9\\
		\hline	
	\end{array}
	\quad
	\begin{array}{|c|}
		\hline
		\dots\\
		\hline	
	\end{array}
	\begin{array}{|c|c|c|}
		\hline
		R&	T	&Pr(T|R)  \\
		\hline
		\true&	\true & 0.7 \\
		\true&	\false	& 0.3\\
		\hline
		\false	&	\true & 0.1 \\
		\false	&	\false	& 0.9\\
		\hline	
	\end{array}
	\]
} 

\condinc{}{
	{\ \[\begin{array}{c  c c}
			\begin{array}{|c|c|}
				\hline
				R	&Pr(R)  \\
				\hline
				\true & 0.2 \\
				\false	& 0.8\\
				\hline
			\end{array}
			\quad
			\begin{array}{|c|c|}
				\hline
				S	&  Pr(S) \\
				\hline
				\true & 0.1 \\
				\false	& 0.9\\
				\hline
			\end{array}
			\quad 
			\begin{array}{|c|c|c|c|}
				\hline
				R	& S  &  W& Pr(W|R,S)\\
				\hline 
				\true	& \true & \true & 0.99 \\
				\true	& \true & \false & 0.01 \\
				\hline
				\true	&  \false & \true  & 0.7\\
				\true	&  \false & \false  & 0.3\\
				\hline
				\false	& \true &  \true& 0.9  \\
				\false	& \true &  \true& 0.1  \\
				\hline
				\false	& \false & \true & 0.01\\
				\false	& \false & \false & 0.99\\
				\hline
			\end{array}
		\end{array}
		\]
	}
	
}


\subsubsection{An example of Proof-Net.}


	Proof-nets~\cite{ll,synsem} are a graphical representation of  Linear Logic sequent calculus proofs.
	The  \BN in \Cref{fig:BNrain} can be encoded in \emph{multiplicative}  linear logic (\MLL) as the  proof-net in  \Cref{fig:ABCDE_D}. 
	
	The nodes $\bbox^A, \bbox^B, \bbox^C, ...$ are  boxes, storing semantical information -- for example, the same conditional probability distributions as in \Cref{fig:BNrain}.
	The  edges of the \pn connect  together and transfers  such information. 
	Notice that edges in \Cref{fig:ABCDE_D} are labelled by atomic formulas, either positive or negative.
	The flow of information in a proof-net (its geometry of interaction~\cite{goi0,goi1,multiplicatives}) follows the polarity of atoms, going downwards on positive atoms (which carry ``output information'')
	and upwards on negative atoms  (``input information'').
	\MLL allows for \emph{duplication of the information carried by negative atoms}---via the   $\cn$-node (contraction). The $\w$-nodes (weakening) \emph{block} 
	the information.
	
	The \pn 
	$\R_D$ in \Cref{fig:ABCDE_D} has a single conclusion $D^+$.
	It represents a marginal probability $\Pr(D)$. 
	Let us informally see how we can draw a sample from $\R_D$.
	\begin{example}[Sampling from a \pn]
		The only  node which is \emph{initial} (w.r.t the flow)  is $\bbox^A$, which   outputs a sample from $\Pr(A)$. 	
		Assume this value is $\true$. This sample is \emph{propagated} via the $\cn$-node (\emph{contraction})  to both $\bbox^C$ and $\bbox^B$. When $\bbox^C$  receives $\true$,  it samples a value from the distribution 
		$\Pr(C|A=\true)$ (for example, $\true$ with probability $0.2$).   Assume the output of $\bbox^C$ is    $\true$, and that the output of   $\bbox^B$ (obtained with a similar procedure) is $\false$. When the box  $\bbox^D$
		receives these values, it  outputs a sample from $\Pr(D|C=\true,B=\false)$, which is a sample from the marginal $\Pr(D)$.
		
		\condinc{}{Notice  that all non initial boxes have  to wait for all their inputs (a form of synchronization \cite{popl17}).}
	\end{example}
	
	\condinc{}{
		We are now  able to informally see how we can draw a sample from a \pn---so clarifying the role of both \emph{weakening} and \emph{contraction}. 
		
		\begin{example}[Sampling from a \pn] Let us consider the \pn  of single conclusion $D$
			in \Cref{fig:ABCDE2}. Assume that each box stores a conditional probability  (take the same CPT's as in \Cref{fig:BNrain}).
			The way information flows in the graph is captured by the 
			polarized order  (as standard in  Linear Logic).
			Nodes which are not initial  need to wait for inputs.
			
			The only initial node  in this \pn is $\bbox^A$, which  outputs  a sample from $\Pr(A)$. 	
			Assume this value is $\true$. This sample is \emph{propagated} via the $\cn$-node (\emph{contraction})  to both $\bbox^C$ and $\bbox^B$. When $\bbox^C$  receives $\true$,  it samples from
			$\Pr(C|B=\true)$ (which according to the CPT in \Cref{fig:BNrain} is $\true$ with probability $0.2$).  Assume the output of $\bbox^C$ is    $\true$, and that the output of   $\bbox^B$ (obtained with a similar procedure) is $\false$. When the box  $\bbox^D$
			receives these values (notice it needs to wait for both inputs), it  outputs a sample from $\Pr(D|C=true,D=false)$, which is a sample from the marginal $\Pr(D)$.
		\end{example}
	}

\SLV{}{	Every \pn is a graph  image of a sequent calculus proof-tree. The following 
	is an example which corresponds to this \pn.  
	{\tiny \[\infer[\cut (A)] {\quad\vdash D^+}
		{\infer[\bbox^A ]{\vdash A^+}{}&
			\infer[\cut(C)/\cn]{\vdash A^-,D^+ }
			{\infer[\bbox^C]{\vdash A^-,C^+}{} & \infer[\cut (B)]{\vdash A^-,C^-,D^+}{\infer[\bbox^B]{\vdash A^-, B^+}{} & 
					{
						\infer[cut(E)]{\vdash B^-,C^-,D^+}{\infer[mix/\cn]{\vdash B^-,C^-,D^+,E^+}
							{\infer[\bbox^D]{\vdash B^-,C^-,D^+}{} & \infer[\bbox^E]{\vdash C^-,E^+}{} }  
							& \infer[\w]{\vdash E^-}{} }
					}
				}
			}
		}
		\]}
	}

	\subsection{The language of Bayesian Reasoning}\label{sec:rvs}
	Let us  briefly revise the language of Bayesian modeling that we use.
	For more details, we refer to~\cite{DarwicheHandbook} for a concise presentation, and to standard texts for an exhaustive treatment~\cite{Pearl88,DarwicheBook,NeapolitanBook}.
	
	Bayesian methods provide a formalism for reasoning about partial beliefs under
	conditions of uncertainty. Since we   cannot determine for
	certain  the state of some features of interest,  we settle for determining how \emph{likely} it is that a particular feature is in a particular state. 
	Random variables (\rvs) represent features  of the system being modeled.   
	A \rv  can be seen as a \emph{name} for an atomic proposition (\eg ``Wet'') which assumes values from a set of states (\eg $\{\true,\false\}$).  The system  is  modeled as \emph{a joint probability distribution} on all possible values of the variables of interest -- each instantiation  representing a possible state of the system.

	\SLV{}{
		\begin{example} The  canonical sample space sketched in \Cref{fig:marginalization}  consists of $2^4$ tuples (only some entries are displayed); to each tuple  $\bx$ 
			is associated a probability.
			The event $(R=\true) $  contains $2^3$ tuples, 
			the event $(R=\true,W=\true)$ contains $2^2$ tuples, and has probability $0.33$.
		\end{example}
	}

	\paragraph{Random Variables.}
	
	We adopt the standard convention of capital letters  (\eg $X,Y$) denoting random variables,   while lowercase letters (\eg $\x,\y$) are particular fixed \emph{values} of those variables, \ie $\x$ an instantiation of the \rv $X$. As standard,  $\Pr (x)$ stands for $\Pr(X=x)$.  
	%
	For simplicity,  random variables are here  taken to be binary, with 
	$\Val{X}=\{\true,\false\}$ -- the generalization to any discrete \rv is straightforward.

	A  finite set of \rvs $\bX=\{X_1, \dots, X_n\}$ defines a ``compound'' \rv whose value set $\Val {\bX}$   is the Cartesian  product 
	$  \Val {X_1} \times \dots \times \Val {X_n} $. A tuple   in the   cartesian product is an instantiation of $\bX$, denoted by  $\xs$.
	%
	%

	\paragraph{Names.}
	Given a countable set $\Names$, we associate to each name $X\in\Names$ 
	a \emph{finite} set of values, denoted by  $\Val {X}$ 
	(typically $\Val{X}=\{\true, \false\}$). From now on, we silently identify a name  $X$ with the pair $(X, \Val {X})$, 
	which effectively defines a  \emph{random variable}.

	\paragraph{Queries to a probabilistic model.}
	Given  a probabilistic model $\Pr(\bX)$, and some variable of interest  $\bY \subset\bX$, a typical query is  the marginal probability  $\Pr(\bY)$ or $\Pr(\bY|\bX=\bz)$, the probability of $\bY$, given  evidence $\bz$ for $\bZ$.
	The former  is obtained by \emph{summing out} the \rvs which are not relevant; the latter
	via  \textbf{Bayes's rule}, which leads us to compute two \textbf{marginals}, $\Pr(\bY, \bZ)$ and ${\Pr(\bZ)}$, because	
	$\Pr(\bY|\bZ)= 
	\frac{\Pr(\bY, \bZ)}{\Pr(\bZ)}  $.
	%

\subsection{Bayesian Networks, formally}
The challenge of Bayesian  reasoning is that 
a joint probability distribution is usually too large to be feasibly represented explicitly. For example, a joint probability distribution over $64$ binary  random variables, corresponds to $2^{64}$ entries.

	\BNs  represents a joint probability distribution in a compact way via a \emph{factorized representation}, obtained by associating with each  node in the DAG a
conditional probability table (\cpt). The semantics of BNs is revised  in \Cref{sec:BN}.

	A \emph{\BN} $\bn$ over the set of \rvs $\bX$ is a pair $(\mathcal G,\aphi)$ where:
	\begin{itemize}
		\item $\mathcal G$ is a directed acyclic graph (DAG) over the set of nodes $\bX$. 
		\item 	$\aphi $   assigns, to each variable $X\in \bX$ a \emph{conditional probability table (a \cpt)} for $X$ given its parents.
	
	\end{itemize}
%


\section{Bayesian Proof-Nets}
Bayesian proof-nets are built on the graph syntax of Multiplicative Linear Logic, extended with probabilistic boxes, which encode   conditional probabilities tables (\CPT's).
We revise the syntax of \MLL \pns,  to define Bayesian proof-nets in \Cref{sec:bpn} and their graph-rewriting rules.

Finally,  we show (\Cref{sec:bpn_BN}) that the correspondence between   Bayesian proof-nets and \BNs sketched in \Cref{fig:BN} and \ref{fig:BNtoPN}   sound and complete.

\subsection{Multiplicative Linear Logic (with probabilistic boxes)}\label{sect:mll}

We assume given a countable set 
of symbols, denoted by metavariables $X, Y, Z$.
The grammar of formulas is that of the multiplicative fragment of linear logic (\MLL):
\begin{equation*}
	\FormA,\FormB ::=
	X^+
	\mid X^-
	\mid\One
	\mid\Bottom
	\mid\FormA\otimes\FormB
	\mid\FormA\parr\FormB
\end{equation*}
We call $X^+$ (\resp\ $ X^-$) a \textdef{positive} (\resp\ \textdef{negative}) \textdef{atomic formula}.
\textdef{Negation} $(\cdot)\orth$ is defined inductively by
$(X^+)^\perp\defeq X^-$,
$( X^-)^\perp\defeq X^+$, 
$\One^\perp \defeq\Bottom$,
$\Bottom^\perp \defeq\One$,
$(\FormA\otimes\FormB)^\perp \defeq\FormA^\perp\parr\FormB^\perp$ and
$(\FormA\parr\FormB)^\perp \defeq\FormA^\perp\otimes\FormB^\perp$.

A \textdef{sequent} is a finite sequence $\FormA_1,\dots, \FormA_n$ of formulas.
Capital Greek letters $\Gamma,\Delta,\dots$ vary over sequents.

\paragraph{Calculus.}
The calculus $\MLLB$ is an extension of multiplicative Linear Logic $\MLL$, introduced in~\cite{EhrhardFP23} (and similar to~\cite{GuerriniM01}).
Beside the \MLL rules, there is a rule for a generalized axiom, called $\boxn$ (which we discuss  later):
\[\infer[\boxn]{\vdash \Yn_1, \dots, \Yn_k, \X }{}  \quad (k\geq 0)\]
All formulas in the conclusion of a $\boxn$-rule are atomic, with exactly one -- the \textdef{main} conclusion $\X$ -- being positive and such that $X \neq Y_i~(\forall i \leq k)$.

In Linear Logic, proofs admit both a sequent calculus syntax, in the form of \textdef{\pts} (\ie trees of sequent calculus rules), and a graph syntax, called \textdef{\pn}. 
Here we are interested in the latter, described in \Cref{fig:pn}.
The corresponding sequent calculus is standard -- for completeness' sake its rules are given in \Cref{app:sequent} together with their images as \pns.
The grammar	in \Cref{fig:pn} adds a new kind of nodes to standard \pns, namely $\boxn$, which corresponds to the new $\boxn$-rule.
We also call boxes the $\boxn$-nodes.

\begin{definition}[\MLLB Proof-Net]\label{def:pn}
An \MLLB proof-net is a \emph{typed graph} $\Net$ (point \ref{item:pn:1}) which satisfies conditions \ref{item:pn:2} (pending edges) and \ref{item:pn:3} (correctness).
\begin{enumerate}
\item\label{item:pn:1} \textbf{Typed graph}\footnote{A typed graph is often called a proof-structure or a module in the literature.}.
$\Net$ is a partial\footnote{We admit pending edges.} graph whose edges are labelled by \MLL formulas.
The alphabet of nodes is given in \Cref{fig:pn}: each node is labelled by a rule.
The edges incident to a node $\node$ need to respect the typing conditions given in the grammar, and are classified either as \emph{premises} of $\node$ (depicted above $\node$) or as its \emph{conclusions} (depicted below $\node$).
We further require that an edge is the conclusion of at most one node and the premise of at most one node.
	
\item\label{item:pn:2} \textbf{Pending edges}.
Every edge in $\Net$ is the conclusion of some node\SLV{}{ (but not necessarily a premise)}.
The edges which are premise of no node (\ie\ the pending edges) are called the \emph{conclusions} of $\Net$.
	
\item\label{item:pn:3} \textbf{Correctness}.
Every cycle in $\Net$ uses at least two premises of a same $\parr$-node or of a same $\cn$-node.
\end{enumerate}
\end{definition}

We write $\R:\Delta$ for a proof-net $\R$ of conclusions $\Delta$. 
A \textdef{sub-net} of a \pn $\Net$ is a sub-graph of $\Net$ which is itself a proof-net.

\begin{figure}
\begin{minipage}[t]{0.48\linewidth}
	\fbox{\includegraphics[page=2,width=1\linewidth]{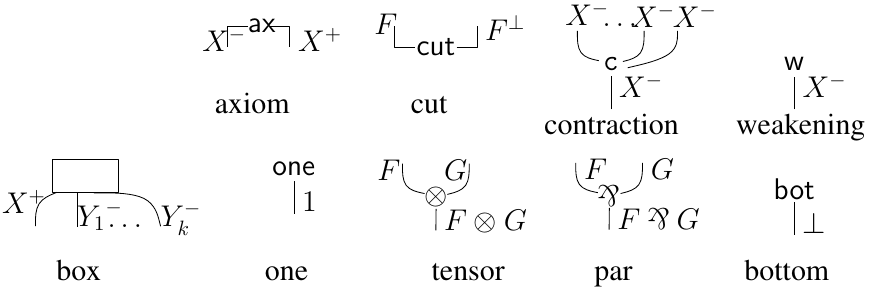}\vspace*{-8pt}}
	\caption{Grammar of nodes. For a $\boxn$-node, $k\geq 0$ and $X \neq Y_i~(\forall i \leq k)$}
	\label{fig:pn}
\end{minipage}
\hfill
\begin{minipage}[t]{0.48\linewidth}
	\fbox{\includegraphics[page=3,width=1\linewidth]{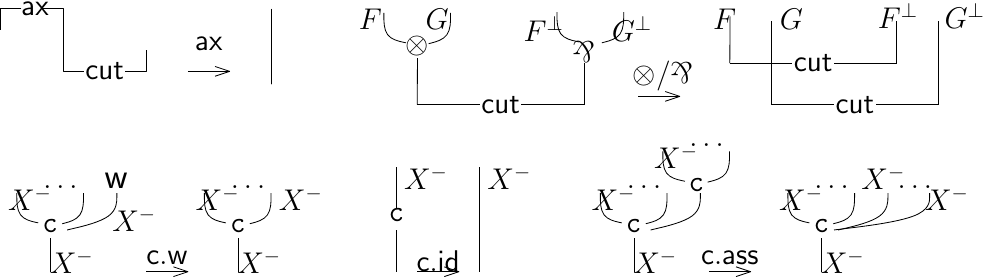}\vspace*{-8pt}}
	\caption{Reduction Rules}
	\label{fig:rewriting}
\end{minipage}
\end{figure}

The correctness condition (point \ref{item:pn:3}) ensures the graph is the image of a \pt\ from sequent calculus.

\begin{theorem}[Sequentialization]\label{th:sequentialisation}
Every proof-net $\Net$ is the image of a \pt\ from sequent calculus, \ie\ it can be inductively generated by the rules in \Cref{fig:sequentialization}.
\end{theorem}

\SLV{}{\begin{example}
The \pt\ in \Cref{fig:ABCDE_seq}\todor{lacking fig} is a sequentialization of the proof-net in \Cref{fig:ABCDE_D}.

\SLV{}{Notice that (as rather standard) we sequentialize in a proof where contractions are greedy, that is they are performed as soon as possible -- so in fact we embed contraction into the rules for $\cut$, $\otimes$, and mix.}
\end{example}}


\subsubsection{Probabilistic boxes: sampling and conditionals.}
At a first glance, the $\boxn$-nodes, to which we will associate a \CPT, seem unusual compared to the usual proof-nets of linear logic.
We sketch here how these boxes can be obtained using the additive connectives of linear logic, with as sole addition a \textit{sample-node} representing a coin toss.

Assume given sample-nodes of the shape
\resizebox*{!}{2em}{
\begin{tikzpicture}[baseline=(b)]
\coordinate (b) at (0,-.5);
\begin{scope}[every node/.style={circle,draw=black}, every path/.style={draw=black}]
	\node (s) at (0,0) {$p$};
	\path (s) -- node[named edge]{$\X$} ++(0,-1);
\end{scope}
\end{tikzpicture}
}
where $0 \leq p \leq 1$.
Such a node represents a probability distribution on the binary \rv $X$: it is true with probability $p$, and false with probability $1-p$ -- see \cref{fig:probabilistic_boxes} for an illustration.
We see a positive atomic formula $\X$ as an alias for a boolean, which is the formula $\One \oplus \One$ in linear logic.
By duality, this implies that $\Xn$ is an alias for $\bot \with \bot$.
The additive  encoding of booleans as $\One \oplus \One$ is standard; the  ``if then else'' function  (that takes a boolean and returns a boolean, so of type $\vdash \bot\with\bot, \One\oplus \One$) is encoded by a $\with$- rule (an additive box), see \cref{fig:probabilistic_boxes}.

Any   $\cpt$ is hence directly encoded by means of suitable  \emph{sample}-nodes  and (possibly nested)\emph{ if-then-else}, in the standard way.
A graphical representation of the internal working of a $\boxn$-node is depicted on \cref{fig:probabilistic_boxes}.

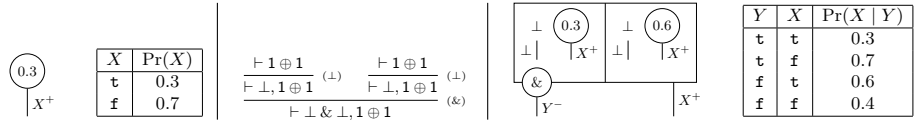
\begin{figure}[t]
\centering
\scalebox{0.6}{
\begin{tikzpicture}
\begin{scope}[every node/.style={circle,draw=black}, every path/.style={draw=black}]
	\node (s) at (0,0) {$0.3$};
	\path (s) -- node[named edge]{$\X$} ++(0,-1);
\end{scope}
\end{tikzpicture}
}
\hspace*{.5em}
\scalebox{0.75}{
\begin{tabular}[b]{|@{\hspace{.5em}}c@{\hspace{.5em}}|@{\hspace{.5em}}c@{\hspace{.5em}}|}
 \hline
 $X$ & $\Pr(X)$ \\
 \hline
 $\true$ & $0.3$ \\
 $\false$ & $0.7$ \\
 \hline
\end{tabular}
}
\vrule\hspace*{.5em}
\scalebox{0.65}{
\begin{prooftree}
\hypo{\One\oplus \One}
\infer1[\bot]{\bot, \One\oplus \One}
\hypo{\One\oplus \One}
\infer1[\bot]{\bot, \One\oplus \One}
\infer2[\with]{\bot\with\bot, \One\oplus \One}
\end{prooftree}}
\hspace*{.5em}\vrule\hspace*{.5em}
\scalebox{0.6}{
\begin{tikzpicture}
\begin{scope}[every node/.style={circle,draw=black}, every path/.style={draw=black}]
	\node (s) at (.25,-.25) {$0.3$};
	\path (s) -- node[named edge]{$\X$} ++(0,-.75);
	\node (s') at (2.25,-.25) {$0.6$};
	\path (s') -- node[named edge]{$\X$} ++(0,-.75);
	\draw[draw=black] (-1,.25) rectangle (3,-1.5);
	\path (1,.25) -- (1,-1.5);
	\node[fill=white] (w) at (-.5,-1.5) {$\with$};
	\node[draw=none] (b) at (-.5,-.25) {$\bot$};
	\path (b) -- node[named edge,left]{$\bot$} ++(0,-.75);
	\node[draw=none] (b') at (1.5,-.25) {$\bot$};
	\path (b') -- node[named edge,left]{$\bot$} ++(0,-.75);
	\path (w) -- node[named edge]{$\Yn$} ++(0,-.75);
	\path (2.5,-1.5) -- node[named edge]{$\X$} ++(0,-.75);
\end{scope}
\end{tikzpicture}
}
\hspace*{.5em}
\scalebox{0.75}{
\begin{tabular}[b]{|@{\hspace{.5em}}c@{\hspace{.5em}}|@{\hspace{.5em}}c@{\hspace{.5em}}|@{\hspace{.5em}}c@{\hspace{.5em}}|}
 \hline
 $Y$ & $X$ & $\Pr(X\mid Y)$ \\
 \hline
 $\true$ & $\true$ & $0.3$ \\
 $\true$ & $\false$ & $0.7$ \\
 $\false$ & $\true$ & $0.6$ \\
 $\false$ & $\false$ & $0.4$ \\
 \hline
\end{tabular}
}
\caption{A sample-node (left), ``if then else'' with a $\with$-rule (center), and a probabilistic box built from sample-nodes and additives (right)}
\label{fig:probabilistic_boxes}
\end{figure}
From now on, we do not explicitly display the internal content of each probabilistic box, but we simply associate to it the corresponding \cpt.

\begin{remark}[Weakening and Contraction]\label{rk:negatives}
The syntax allows structural rules (weakening and contraction) on \emph{negative} atomic formulas.
Recall that $\Xn$ stands for $\bot\with\bot$.
It is well-known in linear logic that formulas built from $\bot$ by means of $\with$ (and $\parr$), \ie formulas of negative polarity, admits weakening and contraction.
\end{remark}

\subsection{Bayesian Proof-Nets}\label{sec:bpn}

\begin{definition}[Bayesian proof-net (\BPN)]
We call Bayesian a proof-net $\R:\Delta$ with probabilistic boxes which satisfies the following two conditions:
\begin{enumerate}
\item 
\emph{the atoms labelling the positive conclusions of the boxes are \emph{pairwise distinct}}; and
\item all the atoms in the conclusions $\Delta$ are positive.
\end{enumerate}
We also call Bayesian a \pn that is a sub-net of a  Bayesian \pn. 
\end{definition}
We write \BPN\ for Bayesian proof-net.
Because of condition (1), we denote each box in a \BPN\ $\R$ by the name of its (unique) positive conclusion $\X$, writing $\bbox^X$.
We call \textdef{positive} a proof-net $\R:\Delta$ respecting  condition (2) -- all atomic (sub-)formulas in $\Delta$ are positive.
We denote by $\Boxes\R$ the \textdef{set of all boxes} of $\R$.
We call \textdef{main names} of $\R$ (noted $\Main{\R}$) the set of names labelling the positive conclusion of boxes in $\R$.
	
\subsection{Reduction, expansion, and  normal forms}\label{sec:pn_rewriting}
	

\Cref{fig:rewriting} sketches the $\cut$-rewriting steps plus the structural rules for contraction and weakening.
All these rules are standard, and preserve both the correctness and the conclusions of a proof-net.

\begin{definition}[Reduction, expansion, and normal forms]\label{def:rewriting}
The normalization rules for \MLLB proof-nets are given in \Cref{fig:rewriting}. 
Each of these rewriting rules $\rr$ defines a binary relation $\Red[\rr]$ on proof-nets, called a \textbf{reduction} step and written $\R\Red[\rr] \R'$ (read $\R$ $r$-reduces to $\R'$).
The inverse step is called an \textbf{expansion}.
So if $\R\Red[\rr] \R'$, then $\R$ is an $r$-expansion of $\R'$.
We write $\Red[\rules]$ for a reduction according to any rule in \Cref{fig:rewriting}. 
$\R$ is in normal form (or just normal) if there is no $\R'$ such that $\R\Red[\rules] \R'$.
\end{definition}

	\SLV{}{	\begin{example}The \pn  in \Cref{fig:ABCDE2} (where   $\R$ is as in \Cref{fig:ABCDE1}) reduces to the \pn of \Cref{fig:ABCDE_D}, which is normal.
		\end{example}
}

\begin{proposition}
The reduction $\Red[\rules]$ is terminating and confluent.
\end{proposition}

\subsubsection{Normal Forms and Atomic \pns.}
Please notice that, because of \paxs, a \pn $\Net$ in normal form can still contain cuts.
However, $\Net$ has a special shape, for it can be decomposed as a \pn whose edges are all labelled by atoms, on top of the formula trees of its conclusions.
		
	\SLV{}{	
		\begin{remark}\label{rk:normal_atomic}
			If 	  $\ProofN$ is a   normal \pn with atomic conclusions, then $\ProofN$ is  atomic.
		\end{remark}
		\begin{lemma}[normal proof-nets ]\label{lem:normal_net}
			Let  $\R:\Delta$ be a  \emph{normal} \pn. Then 
				(1.)  the premises of each $\cut$-node are atomic, 
				(2.) the  positive premise  of each cut-node is  conclusion of an $\Ax$-node. 
		\end{lemma}}
		\condinc{}{	\begin{lemma}[normal proof-nets ]\label{lem:normal_net}
				A  proof-net $\R$  which is \emph{normal} (\ie $m$-normal)  satisfies the following properties: 
				\begin{enumerate}
					\item  the premises of each $\cut$-node are atomic
					\item the  positive premise  of each cut-node is  conclusion of an $\Ax$-node
					\item the negative premise of any  cut-node is conclusion of either an $\Ax$-node or a $@$-node or a $\weak$-node
					\item the conclusion of any $\weak$-node is either conclusion of $\R$, or premise of a cut.
				\end{enumerate}
				
		\end{lemma}}

		\condinc{}{\begin{lemma}
				If $\R$ is $\Red[m,s]$-normal,  it contains no $\weak$-node.
			\end{lemma}
		}

\paragraph{Atomic \pns.}
A \pn $\Net$ is said \textdef{atomic} if all formulas labelling its edges are atoms. 
We write that $e$ is labelled by $X$ if its label is either $X^+$ or $X^-$. 

\begin{property}[Normal forms]\label{fact:canonical_dec}
Let $\Net:F$ be a \emph{normal} \pn of conclusion $F$.
Then $\N$ can be decomposed in two subgraphs $\Net^\at$ and $\Net^F$, where $\Net^\at:\Gamma$ is an atomic sub-net whose conclusions $\Gamma$ is the sequence of the atomic subformula of $F$, while $\Net^F$ is a typed graph with premises $\Gamma$, consisting of the syntax tree of $F$.
A similar property holds if $\Net$ has conclusions $F_1, \dots, F_k$.
\end{property}

\SLV{}{\begin{example}
In \Cref{fig:ABCDE1}\todor{lacking fig}, the syntax tree of the (only) conclusion $C^+\otimes D^+$  is red-colored.	
\end{example}}
	



\subsection{Bayesian Networks and Bayesian Proof-Nets}\label{sec:bpn_BN}

We prove that every \BN is associated to a positive \Bpn and vice versa.
To formalize this result, we need the notion of polarized order.

\paragraph{The polarized order.}\label{sec:polarized}

Let $\R$ be an atomic \pn (which is a simple case of \emph{polarized} proof-nets~\cite{Laurent03}). 
We write $\pol{\R}$ the graph which has the same nodes and edges as $\R$, but where the orientation follows the polarity of the labels: downwards if positive, upwards if negative. 	
The following easy property is well known for all polarized proof-nets.

\begin{lemma}[Polarized correctness~\cite{phdlaurent}]\label{lem:pol_correct}
Assume that $\R$ is a typed graph (\Cref{def:pn}, point \ref{item:pn:1}) whose edges are all labelled by atoms.
Then $\R$ is correct if and only if $\pol{\R}$ is a DAG.
\end{lemma}

\paragraph{Polarized order.}
The polarized orientation induces a partial order $<_\text{pol}$ on the nodes of $\R$.
In particular, if $\bbox_1$ and $\bbox_2$ are \paxs, $\bbox_1<_\text{pol}\bbox_2$ if there is a path in $\pol{\R}$ from $\bbox_1$ to $\bbox_2$.
We define $\Dag\R$ to be the DAG associated to $(\Boxes{\R}, <_\text{pol})$. 

				\SLV{}{We refer to $ \bnet{\R} $ also as the \emph{core invariant} of $\R$. 
		The set $\Box(\R)$ as well as  the induced partial  order on it is an  \emph{invariant} of the proof-net under several 
		important transformations, and in particular under normalization.

		\begin{proposition}[Core invariant] Let $\bnet \R= (\Ax{(\R)}, <_\text{pol})$, where 
			$ \bbox_1<_\text{pol}\bbox_2 $ if there is a polarized path  from $ \bbox_1$ to $ \bbox_2$.
			
			$\bnet{\R}$ is invariant under $\arule$-reductions and $\arule$-expansions.
		\end{proposition}
	}

\subsubsection{From Bayesian Networks to \Bpns.}
To every \BN $\bn$ is is straightforward to associated a \pn $\Net_\bn$, as sketched in \Cref{fig:BNtoPN}.
The only delicate point is to check that the typed graph produced by the translation is correct.
This is a consequence of \Cref{lem:pol_correct}, the fact that $\Net_\bn$ is atomic, and that $\pol{\Net_\bn}$ is a DAG, because $\bn$ is.
That $\Net_\bn$ is Bayesian is immediate.
		
\subsubsection{From \Bpns to Bayesian Networks.}
Remarkably, we can associate a \BN $\bnet{\R}$ to the normal form of any positive \BPN $\R$.
Let $\Net$ be its normal form; recalling \cref{fact:canonical_dec}, let $\Net^\at$ be the largest atomic sub-net of $\Net$.
Since each box has exactly one main name, there is an immediate correspondence between $(\Boxes{\Net}, <_\text{pol})$ and the DAG of a \BN over $\Main{\R}=\{X_1, \dots,X_n\}$.

	\condinc{}{	Summing up:
		\begin{proposition}
			\begin{enumerate}
				\item 	Let   $\R:\Gamma$ be a  Bayesian \pn in normal form,  and  $\Main{\R}=\{X_1, \dots,X_n\}$. Then	 
				$  (\Ax{(\R)}, <_\text{pol})$ is isomorphic to a  DAG  $\netG$ over   $X_1, \dots,X_n$. 
				
				\item Conversely, given a Bayesian Network   $(\netG, \Phi)$ over   $X_1, \dots,X_n$, there exists   an atomic 
				\BPN $\R$ such that   
				$  (\Ax{(\R)}, <_\text{pol})$ is isomorphic to $\netG$. 
			\end{enumerate}
		\end{proposition}
	}
\begin{example}
Consider the \pn $\R$ in \Cref{fig:BNtoPN}
It is clear that the polarized order  exactly describes  the DAG of the BN in \Cref{fig:BNrain}.
\end{example}

Notice that the \BN  associated to a \BPN $\R$ is independent from the conclusions of $\R$.

\subsubsection{Queries.}
A fine point in the correspondence between BN's and \pns deserves discussion.
A positive \BPN $\R:\Gamma$ encodes both a \BN (the probabilistic model) and a \emph{query} to the model: the query is expressed by the conclusions $\Gamma$.
As will be formalized in \Cref{sec:sem}, $\R:\Gamma$ corresponds to a \BN $\bn_\R$ over $\Nm{\R}$.
While $\bn_\R$ defines a probability distribution $\Pr$ over $\Nm{\R}$, the semantics of $\R:\Gamma$ is the \emph{marginal distribution} of $\Pr$ over $\Nm{\Gamma} \subseteq \Nm{\R}$.


\section{Graphs Decomposition}
\label{sec:decomposition}

Trees (such as  proofs, terms, and  type derivations) decompose accordingly to their inductive structure. 
A \pn being the image of a proof-tree $\pi$, the decomposition reflecting $\pi$ is always possible. However, the fact that a \BPN 
is a graph, it offers much more flexibility in term of decompositions--- we are  more  free in the choice of the  subgraphs.

Still, there are some subtleties. When factorizing a \BPN $\N$, we wish to decompose it in  several sub-nets connected by (any number of) $\cut$-nodes. That may not be immediately possible for $\N$, but the issue is easily solved by  first performing an \emph{expansion} of $\Net$, as described below.

\begin{remark}[Cutting edges]
Observe that $\ax$-expansion (the reverse of $\ax$-reduction) introduces a $\cut$-node, which allows to \emph{split}  the edge in two parts.
\end{remark}

\begin{definition}[Splitting]\label{def:decomposition}
We write $\N= \Cuts{\R,\netS}$ if $\N$ can be decomposed in two sub-nets $\R$ and $\netS$ connected by (any number of) $\cut$-nodes.

Let $\N$ be a proof-net, and $\R \subseteq \N$ an atomic sub-net.
The \textbf{splitting of $\N$ induced by $\R$} is performed by replacing with its $\ax$-expansion each edge which is both a conclusion of $\R$ and a premise of a node outside $\R$.
We call the resulting \pn $\pulledout{\N}{\R}$.
It is an expansion of $\N$ which decomposes as $\pulledout{\N}{\R} = \Cuts{\R,\netS}$, where $\netS$ is the sub-net of $\pulledout{\N}{\R}$ obtained by removing $\R$ and the newly introduced $\cut$-nodes.
(See \cref{fig:pullingout} for an illustration.)
\end{definition}

\SLV{}{We can generalize this construction to the case where $\R$ is not atomic, by means of generalized $\ax$-nodes, that are derivable from the $\ax$-nodes we have.}


\begin{figure}
\begin{adjustbox}{max height=5em,center=\linewidth}
\begin{tikzpicture}
\node at (0,-.7) {$\R$};
\draw[dashed] (0,-.7) ellipse (10mm and 5mm);
\draw[dashed] (0,-2.3) ellipse (18mm and 6mm);
\begin{scope}[every path/.style={draw=black}]
	\path (.75,-.85) -- (1,-1.5);
	\path (.25,-1) -- (.25,-1.9);
	\path (-.25,-1) -- (-.25,-1.9);
\end{scope}
\end{tikzpicture}
\hspace*{1em}\vrule\hspace*{1em}
\begin{tikzpicture}
\node at (0,-.7) {$\R$};
\draw[dashed] (0,-.7) ellipse (10mm and 5mm);
\draw[dashed] (-5,-2.3) ellipse (18mm and 6mm);
\node at (-5.8,-.9) {$\color{red}\netS$};
\draw[dotted,red] (-4.7,-1.75) ellipse (25mm and 13mm);
\begin{scope}[every node/.style={rectangle,draw=none}, every path/.style={draw=black}]
	\path (.75,-.85) -- (1,-1.5);
	\node (ax1) at (-3.5,-.75) {$\ax$};
	\node (ax2) at (-3.5,-1.25) {$\ax$};
	\node (cut1) at (-1,-2) {$\cut$};
	\node (cut2) at (-1,-2.5) {$\cut$};
	\path (-.25,-1) |- (cut1);
	\path (cut1) -- ++ (-.5,0) |- (ax1);
	\path (ax1) -| (-5.25,-1.9);
	\path (.25,-1) |- (cut2);
	\path (cut2) -- ++ (-1,0) |- (ax2);
	\path (ax2) -| (-4.75,-1.9);
\end{scope}
\end{tikzpicture}
\end{adjustbox}
\caption{Illustration of \cref{def:decomposition} with $\N$ on the left and $\pulledout{\N}{\R}$ on the right}\label{fig:pullingout}
\end{figure}
\SLV{}{
\begin{example}\label{ex:pullingout}
Consider the \pn $\N$ on the left of \cref{fig:pullingout_ex}, with $\R$ made of $\boxn^Y$, $\boxn^Z$ and the rightmost $\cut$.
The resulting $\pulledout{\N}{\R}$ is depicted on the right of \cref{fig:pullingout_ex}.
\end{example}
	
\begin{figure}
	\begin{adjustbox}{}
		\begin{tikzpicture}
			\begin{scope}[every node/.style={rectangle,draw=none}, every path/.style={draw=black}]
				\node[draw=black] (bw) at (-1.6,-1) {$\boxn^W$};
				\node[draw=black] (bx) at (.5,-1) {$\boxn^X$};
				\node[draw=black] (by) at (3,0) {$\boxn^Y$};
				\node[draw=black] (bz) at (6,0) {$\boxn^Z$};
				
				\node (cw) at (-.6,-3.75) {$\cut$};
				\node (aw) at (.1,-3.25) {$\cn$};
				\node (cy) at (4.5,-2) {$\cut$};
				\path (-1.3,-1.3) |- node[named edge,pos=.2,left]{$W^+$} (cw) -| (aw);
				\path[out=90,in=-90] (aw) edge node[named edge]{$W^-$} (.1,-1.3);
				\path[out=75,in=-90] (aw) edge node[named edge,pos=.9,left]{$W^-$} (2.6,-.3);
				\path[out=60,in=-90] (aw) edge node[named edge,pos=.9]{$W^-$} (5.6,-.3);
				\path (.8,-1.3) -- node[named edge,pos=.75]{$\X$} ++(0,-2.2);
				\path (3.3,-.3) |- node[named edge,pos=.15]{$\Y$} (cy) -| node[named edge,pos=.25]{$\Yn$} (5.8,-.3);
				\path (6.3,-.3) -- node[named edge,pos=.75]{$\Z$} ++(0,-1.7);
			\end{scope}
			\draw[dashed,red] (4.5,-.8) ellipse (27mm and 17mm);
			\node at (4.5,.5) {\color{red}$\R$};
		\end{tikzpicture}
		\hspace*{1em}\vrule\hspace*{1em}
		\begin{tikzpicture}
			\begin{scope}[every node/.style={rectangle,draw=none}, every path/.style={draw=black}]
				\node[draw=black] (bw) at (-4.6,-1) {$\boxn^W$};
				\node[draw=black] (bx) at (-2.5,-1) {$\boxn^X$};
				\node[draw=black] (by) at (3,0) {$\boxn^Y$};
				\node[draw=black] (bz) at (6,0) {$\boxn^Z$};
				
				\node (cw) at (-3.6,-3.75) {$\cut$};
				\node (aw) at (-2.9,-3.25) {$\cn$};
				\node (cy) at (4.5,-2) {$\cut$};
				\node (a1) at (-.75,.1) {$\ax$};
				\node (c1) at (1.25,-.9) {$\cut$};
				\node (a2) at (-.75,-.4) {$\ax$};
				\node (c2) at (1.25,-1.4) {$\cut$};
				\path (-4.3,-1.3) |- node[named edge,pos=.2,left]{$W^+$} (cw) -| (aw);
				\path[out=90,in=-90] (aw) edge node[named edge]{$W^-$} (-2.9,-1.3);
				\path[out=45,in=-90] (aw) edge (-1.75,-.4);
				\path (-1.75,-.4) |- (a1);
				\path (a1) -| (.25,-.9) -- (c1);
				\path (c1) -| node[named edge,pos=.9]{$W^-$} (2.6,-.3);
				\path[out=25,in=-90] (aw) edge (-1.5,-.9);
				\path (-1.5,-.9) |- (a2);
				\path (a2) -| (0,-1.4) -- (c2);
				\path (c2) -| node[named edge,pos=.9]{$W^-$} (5.6,-.3);
				\path (-2.2,-1.3) -- node[named edge,pos=.75]{$\X$} ++(0,-2.2);
				\path (3.3,-.3) |- node[named edge,pos=.15]{$\Y$} (cy) -| node[named edge,pos=.25]{$\Yn$} (5.8,-.3);
				\path (6.3,-.3) -- node[named edge,pos=.75]{$\Z$} ++(0,-1.7);
			\end{scope}
			\draw[dashed,red] (4.5,-.8) ellipse (27mm and 17mm);
			\node at (4.5,.5) {\color{red}$\R$};
		\end{tikzpicture}
	\end{adjustbox}
\caption{Example of splitting with $\N$ on the left and $\pulledout{\N}{\R}$ on the right}\label{fig:pullingout_ex}
\end{figure}
}
Whereas such a decomposition is very natural from a graphical perspective, it has no direct correspondence in the inductive syntax of sequent calculus.
In particular, while $\pulledout{\N}{\R}$, $\R$ and $\netS$ all have corresponding \pts, there is no reason that the \pt\ associated to $\pulledout{\N}{\R}$ can be obtained from those of $\R$ and $\netS$.
Indeed, there is no rule in sequent calculus allowing to compose two given \pts\ by multiple cuts.
Thus, more decompositions are allowed in \pns\ than in \pts.

\section{Semantics and Compositionality}


\newcommand{\Cpts}[1]{\mathsf{Cpts}(#1)}
\newcommand{\dem}{:}
\newcommand{\set}[1]{\{#1\}}
\newcommand{\tm}{t}
\newcommand{\iL}{L}

Following \cite{popl24}, we adopt the same semantics as that of \BNs. After recalling the  main ingredients, we  extend the proof of compositionality of \cite{popl24}, by adopting  a notion  of ``component'' which is \emph{specific to graphs}, and hence does not appear in previous work.

\subsection{Factors, and the semantics of \BNs }

%
%
%
%
%
%
%
%

Inference algorithms  rely on basic operations on a class of functions 
known as \emph{factors}, which generalize    the notions of  probability distribution and of conditional distribution. Factors will be the key ingredients also in our semantics.

\begin{definition}[Factor]
	A \emph{factor} $\ft\phi{\bX}$  over a set of   \rvs $\bX$
	is a function $\ft\phi{\bX}:\Val{\bX}\to\mathbb{R}_{\geq 0}$
	mapping each  tuple  $\bx\in \bX$  to a non-negative real.
\end{definition}
 Letters $\phi,\psi$ range over factors.
When $\bX$ is clear from the context, we simply write  $\phi$ (omitting the superscript $\bX$); then $\Nm{\phi}$  denotes $\bX$.
 In the literature about BNs, $\phi(\bx)$ is often written $\phi_\bx$. We adopt this  convenient notation in explicit calculations.
\condinc{}{
\begin{example} Factors generalize familiar concepts from probability theory.
\begin{itemize}
	\item A \emph{joint probability distribution} over the set $\bX$ is a factor $\phi$ which maps each tuple  $\xs\in \Val \bX$ to
	a probability $\phi(\xs)$ such that $\sum_{\xs\in \Val{\bX}} \phi(\xs) =1$.
	\item A \emph{conditional probability table} (\cpt)  for $X$ given $\bY$  is a factor $\ft\phi {\{X\}\cup\bY}$ which maps each tuple $\x \ys\in\Val{\{X\}\cup\bY}$ to a probability $\phi(\x \ys)$ such that for each $\ys\in\bY$, $\sum_{\x\in \Val{X}} \phi(\x \ys) =1$.\footnote{Please notice that here, and in the following definition of sum out, we slightly abuse the notation. In fact, as standard, we consider the tuples as sequences indexed by the set of random variables. 
	Every time, we present the tuples ordered in the most convenient way for a compact  definition.}
\end{itemize}
\end{example}
}
Factors come with two important operations: sum (out) and product.  Product  of factors is defined in such a way  that only ``compatible`` instantiations are  multiplied.

	Given a subset $\bY\subseteq\bX$, we denote by $\Proj{\bx}{\bY}$ the restriction of $\bx$ to ${\bY}$ (so, $\Proj{\bx}{\bY} \in \Val \bY$ ). 
	Given two sets of names $\bX$ and $\bY$, the instantiations  $\bx\in\Val{\bX}$ and $\by\in\Val{\bY}$  \emph{are compatible } ($\bx\sim \by$ for short) whenever $\Proj{\bx}{\bX\cap\bY}=\Proj{\by}{\bX\cap\bY}$, \ie they agree on the common names.

%
%
\begin{itemize}
\item The \textbf{sum out} of $\bZ\subseteq\bX$ from  $\ft\phi{\bX}$ is a factor $\sum_\bZ \phi$ over  $ \bY \defeq \bX  - \bZ $, defined as: 
	\[\left(\sum_\bZ \phi\right)(\by)\defeq\sum\limits_{\bz\in \Val{\bZ}}\phi(\bz\, \by)\]

\item	The \textbf{product} of $\ft{\phi_1}{\bX}$ and $\ft{\phi_2}{\bY}$  is a factor $\phi_1\FProd \phi_2$ over $\bZ \defeq \bX\cup \bY$, defined as:
	\[(\phi_1\FProd \phi_2)(\bz) \defeq \phi_1(\bx) \phi_2(\by) \quad  \mbox{where $\bx= \Proj{\bz}{\bX}$ and $\by= \Proj{\bz}{\bY}$.}\]
\condinc{}{\RED{where $\bx\sim \bz$ and $\by\sim \bz$, \ie $\bx$ and $\by$ 
agree with $\bz$ on  the common variables (see  \Cref{notation_BN}).}}
\end{itemize}
We denote $n$-ary products by $\BigFProd_{n} \phi_n $. We denote by $\ftone_\bY\defeq\ft\ftone\bY$ the factor over the set of names $\bY$, sending every tuple of $\Val{\bY}$ to $1$. Observe that $\ft{\phi}{\bX} \odot \ft{\ftone}{\bY} = \ft{\phi}{\bX}$ if $\bY \subseteq \bX$.
Factors over an empty set of variables are allowed, and called trivial. In particular, we write $ \ftone_\emptyset\defeq\ft \ftone \emptyset$ for 
the trivial factor assigning $1$ to the empty tuple.
%
Product and summation  are both commutative, product is associative, and---crucially---they distribute \emph{under suitable conditions}: 
\begin{equation} \label{eq:sum_prod}
	\text{If }\bZ\cap \Nm{\phi_1}=\emptyset \text{ then } 
	\sum_\bZ ( \phi_1 \FProd \phi_2)  = \phi_1 \FProd \left( \sum_\bZ  \phi_2 \right)  
\end{equation}
This distributivity is  the key property  on which  exact inference algorithms rely.
CPT's being factors, they admit sum and product operations. Please notice that the result of such operations is not necessarily a \cpt, but, of course, it is a factor. 
\begin{remark}[Cost of  Operations on Factors]\label{rem:factors_cost}
	Summing out any number of variables from a factor $\phi$ demands $\BigO{2^{w}}$ time and space, with $w$ the number of variables over which $\phi$ is defined.
Multiplying $k$ factors requires 
$\BigO{k\cdot 2^{w}}$
time and space, where $w$ is the number of variables in the resulting factor.
\end{remark}

\subsubsection{The Semantics of Bayesian Networks.}\label{sec:BN}

%
%
%
The \cpt assigned to each node of a \BN is   a  factor $ \phi^X$ over variables $\{X\}\cup \Pa X$, where $\Pa X$  denotes the set of  parents of $X$ in $\mathcal G$. Independence assumptions on BNs yield:
\begin{theorem}[\cite{Pearl86}]\label{thm:pearl}
	A \BN $\bn$ over the set of \rvs $\bX$ defines a unique probability distribution over $\bX$ (its \emph{semantics}):
	$\sem{\bn} \defeq \BigFProd_{X\in \bX} \phi^X$.
\end{theorem}

Notice that, 
 the  \emph{marginal distribution} of $\sem \bn$ over a subset $\bY \subseteq \bX$ is 
 $\sum_{\bX-\bY} \sem \bn$.

\subsection{Semantics of Bayesian \pns}\label{sec:sem}
The semantics  of a \BPN $\R:\Delta$ is given in a similar way to that of a \BN. 

\paragraph{Semantics of boxes.}
We have seen  that 
to each (name of an) atom we associate  a boolean \rv,  and  to each box $\bbox^X: \X,\Yn_1, \dots\Yn_k$  we associate  
$\Pr(X| Y_1, \dots, Y_k)$,  which is a factor, from now on simply noted $\phi^X$.  The semantics of $\bbox^X$ is hence clear:
\[\sem {\bbox^X} \defeq \phi^X\]

\paragraph{Semantics of a \BPN.}
A positive  \pn $\R:\Delta$ defines a marginal distribution over $\Nm{\Delta}$.
The product of the \cpts associated to the boxes  is the joint probability distribution of the  underlying model; 
finally,  the sum out  all the names not appearing in the conclusion provides the desired   marginal distribution.

%
%

\begin{definition}[Semantics of \BPN's]\label{def:sem}\label{thm:PCoh}
	Let $\R :  \Delta$ be a \BPN.   Its  semantics is 
	\[
	\sem{\R} \defeq \sum_{\bX-\Nm{\Delta}} \left(\bigodot_{\bbox\in \Boxes{\pi}} \sem{\bbox} \right)
	\]
	for  $\bX=\bigcup_{\bbox\in \Boxes{\pi}} \Nm{\bbox}$.
\end{definition}

\begin{remark}\label{ex:var_ax}
	If $\R$ is a \pn such that $\Boxes{\pi}=\emptyset$, then $\sem{\pi}={\ftone}_{\emptyset}$. 
\end{remark}

\subsection{Invariance and compositionality}
Such a  definition of the semantics has two  clear advantages: 
\begin{itemize}
	\item the  correspondence with \BNs semantics is immediate, yielding an efficient representation of the probabilistic model;
	\item the invariance of the semantics via reduction is immediate, because both the boxes and the conclusions are invariant. 
\end{itemize}
The crucial  question however  is \emph{compositionality}. 
It is common in denotational semantics to define the interpretation inductively; compositionality is then intrinsic, while the difficulty  is to prove that the semantics is invariant via reduction and expansion. 
By adopting a definition which is global, invariance is for free, but compositionality needs to be proved. 

Is such a  semantics  compatible with  a modular definition, in terms of  components? 
The answer is positive, and in a stronger sense then in previous work~\cite{popl24}, where compositionality follows the inductive, tree-like definition of the type derivation.

\paragraph{Proof-trees compositionality.} By \Cref{th:sequentialisation}, each \pn is the image of a sequent-calculus \pt, as summarized in \Cref{fig:sequentialization}. Adapting  \cite{popl24}, 
given a \pn $\pi$ obtained from the composition of $\pi_1,\ldots,\pi_n$, 
we  obtain $\sem{\pi}$ from $\sem{\pi_1},\ldots,\sem{\pi_n}$, as follows
\begin{equation}\label{eq:compositionality}
	\sem{\pi} = \sum_{\bZ}  \left( {\BigFProd_i \sem{\pi_i}}\right)\quad
	\text{ where } \bZ= \bigcup_i \Nm{\sem{\pi_i}} - \Nm{\Delta}
\end{equation}
Composition  is obtained by \emph{first} performing the product  $ \BigFProd_i\sem {\pi_i} $---which yields a factor  over the names $\bigcup_i \Nm{\sem {\pi_i}}$--- and \emph{then} marginalizing, by summing out the names which do not appear in the conclusion $\Delta$. Such a notion of composition can be thought of as 
	\emph{composition = parallel composition + hiding.}

Notice that---as stressed in~\cite{popl24}, equation~\ref{eq:compositionality} is not trivial,  because sum out and product do not distribute in general, but only under suitable conditions.

\paragraph{Graph compositionality.}
In this paper, we demand a notion of compositionality which is stronger than in  \cite{popl24}, because  we wish to adopt a more general notion of ``component'', 
which is (only) natural  in a graph-theoretical  setting. As seen above,  since a \pn $\N$ is the image of a \pt, it  can surely be inductively decomposed  following the the rules in table  \Cref{fig:sequentialization}. However, a \pn is first of all a graph, which we may want to manipulate  in terms of \emph{arbitrary} sub-graphs. 
We prove that the semantic of a \pn $\R$ can be defined compositionally whatever is the decomposition of the \pn in sub-nets. 
 In particular, as in \Cref{def:decomposition},  we may decompose a \pn $\R$ in two\emph{ components which are connected by an arbitrary number of cuts}.  
Such a decomposition has no correspondence in the inductive world of type derivations  (\cite{popl24}), or of  sequent calculus (where two \pt can be connected at most by a  \emph{single} cut) 

In \Cref{sec:typing_cut_net} we will prove that we can always type such a decomposition in such a way that the semantics is invariant, and the chosen decomposition can be inductively defined.

\subsection{Graphs compositionality}
The proof of graph compositionality relies on the fact that 
atomic \BPN's satisfy a crucial  property  which makes them akin to jointrees, the data structure underlying message passing, the most used 
 algorithm for exact inference
on \BNs. Namely,  in each  atomic \BPN 
the edges labelled by a  same atom $X$ form a tree, hence any two edges labelled by  the same atom $X$ are connected  by a path in which all  edges have label  $X$.
\begin{lemma}[Jointree-like property] \label{lem:jointree}
	Let $\R : \Lambda$ be a positive  \BPN. 
	\begin{itemize}
		\item The restriction of $\R$ to the  edges labelled by a  same atom  $X$ is a \emph{directed} tree, with root the main conclusion of the box  $\bbox^X$.
		\item Every $X\in\Nm {\R} $ is introduced by a box $\bbox^X$. 
	\end{itemize}
	
\end{lemma}
The proof is in \cref{sec:proofs_sem}. As a  consequence, the following holds for \emph{every} \BPN.
\begin{corollary}[Internal names]\label{cor:internal} 
	Let $\N$ be a \BPN such that   $\N=\Cuts{\N_1, \N_2}$, with $\N_i:\Gamma_i $. Then 
	\[	 Z\in \Nm {\N_j} \tand Z\not \in \Nm {\Gamma_j}\qquad\Rightarrow \qquad Z\not \in \Nm{\N_i} \mbox{ for  } i\not = j.\]	
\end{corollary}

Using this fact, we are able to prove that the semantics can  indeed be computed compositionally, validating \Cref{eq:compositionality}, also for   a graph-theoretical notion of component.

\begin{theorem}[Graphs compositionality]
\label{th:compositional}
Let $\R=\Cuts{\R_1,\R_2}:\Delta$ be a \BPN. Then
$\sem {\R} = \sum_{\bZ}  \left(  \sem {\R_1}  \odot \sem {\R_2}\right )$
where $\bZ=  \big(\Nm{\sem{\R_1}} \cup \Nm{\sem{\R_2}}  \big) - \Nm{\Delta}$.
\end{theorem}

The proof is in \cref{sec:proofs_sem}.

\section{Weakening:  pruning the graph}\label{sec:pruning}

Recall that, given a \BN $\bn$ defining a probability distribution $\Pr$ over $\bX$, we obtain 
the  \emph{marginal distribution} of $\Pr$ over the variables of interest $\bW \subseteq \bX$ 
by summing out the other variables. 
If we are interested only in the subset $\bW$, it is possible to work with a smaller $\bn$ without loosing information, by \emph{pruning} nodes corresponding to \rvs which will not be used: that is, restricting $\bn$ to ancestors of $\bW$.

On the logical side, it is well  understood (and we used extensively) that   marginalization is logically obtained by weakening.
For example in \cref{fig:ABCDE_D}, to marginalize  $E$ it suffices to cut the edge with a   $\weak$-node.
Pruning corresponds to \emph{actually performing} the standard box/$\weak$ reduction described on \cref{fig:w_red}.
This new rewriting system is well-behaved with respect to previous reductions as well as the semantics. 


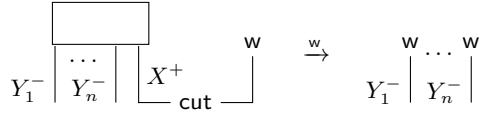
\begin{figure}[t]
	\begin{adjustbox}{}
		\begin{tikzpicture}[baseline=(b)]
			\coordinate (b) at (0,-.5);
			\begin{scope}[every node/.style={rectangle,draw=none}, every path/.style={draw=black}]
				\node[draw=black] (bx) at (0,0) {\phantom{$\boxn^Xpp$}};
				\node (cx) at (1.25,-1.05) {$\cut$};
				\node (w) at (2,-.25) {$\weak$};

				\path (-.6,-1.05) -- node[named edge,left,pos=.25]{$\Yn_1$} ++(0,.75);
				\node at (-.2,-.5) {$\dots$};
				\path (.2,-1.05) -- node[named edge,left,pos=.25]{$\Yn_n$} ++(0,.75);
				\path (.5,-.3) -- node[named edge]{$\X$} ++(0,-.75) -- (cx) -| (w);
			\end{scope}
		\end{tikzpicture}
		$\quad\Red[\weak]\quad$
		\begin{tikzpicture}[baseline=(b)]
			\coordinate (b) at (0,-.5);
			\begin{scope}[every node/.style={rectangle,draw=none}, every path/.style={draw=black}]
				\node (w1) at (-.6,-.25) {$\weak$};
				\node (wn) at (.2,-.25) {$\weak$};

				\path (-.6,-1.05) -- node[named edge,left,pos=.25]{$\Yn_1$} (w1);
				\node at (-.2,-.4) {$\dots$};
				\path (.2,-1.05) -- node[named edge,left,pos=.25]{$\Yn_n$} (wn);
			\end{scope}
		\end{tikzpicture}
	\end{adjustbox}
	\caption{box/$\weak$ reduction}
	\label{fig:w_red}
\end{figure}

\begin{lemma}\label{lem:weak_convergent}
Reduction $\Red[\weak]$ is terminating and confluent.
What is more, $\Red[\rules]\cup\Red[\weak]$ is terminating and confluent.
\end{lemma}

\begin{lemma}\label{lem:weak_red_sem}
Given a positive \BPN\ $\netR$, if $\netR\Red[\weak]\netS$ then $\sem{\netR}=\sem{\netS}$.
\end{lemma}

\Cref{lem:weak_convergent} is immediate, while the proof of \cref{lem:weak_red_sem} is in \cref{sec:proof_pruning}.

\newcommand*{\figR}{\resizebox*{!}{2em}{
\begin{tikzpicture}[baseline=(b)]
\coordinate (b) at (-4.25,-1);
\node at (-4.25,-.6) {$\R$};
\draw[dashed] (-4.25,-.6) ellipse (10mm and 4mm);
\begin{scope}[every node/.style={rectangle,draw=none}, every path/.style={draw=black}]
	\path (-3.5,-.75) -- ++(0,-.5);
	\node at (-3.5, -1.45) {$\Z_1$};
	\path (-4,-.75) -- ++(0,-.5);
	\node at (-4, -1.45) {$\Z_2$};
	\path (-4.5,-.75) - ++(0,-.5);
	\node at (-4.5, -1.45) {$\Z_2$};
	\path (-5,-.75) -- ++(0,-.5);
	\node at (-5, -1.45) {$\X$};
\end{scope}
\end{tikzpicture}}}
\newcommand*{\figRp}{\resizebox*{!}{2em}{
\begin{tikzpicture}[baseline=(b)]
\coordinate (b) at (-4.25,-1);
\node at (-4.25,-.6) {$\R$};
\draw[dashed] (-4.25,-.6) ellipse (10mm and 4mm);
\begin{scope}[every node/.style={rectangle,draw=none}, every path/.style={draw=black}]
	\path (-3.5,-.75) -- ++(0,-.5);
	\node at (-3.5, -1.45) {$\Z_1$};
	\path (-4,-.75) -- ++(0,-.5);
	\node at (-4, -1.45) {$\Z_2$};
	\path (-4.5,-.75) - ++(0,-.5);
	\node at (-4.5, -1.45) {$\Z_2$};
	\node (cutX1) at (-5.75,-1.25) {$\cut$};
	\node (wX1) at (-6.5,-.75) {$\weak$};
	\path (-5,-.75) |- node[named edge,above left]{$\X$} (cutX1);
	\path (wX1) |- node[named edge,above right]{$\Xn$} (cutX1);
\end{scope}
\end{tikzpicture}}}
\newcommand*{\figS}{\resizebox*{!}{2em}{
\begin{tikzpicture}[baseline=(b)]
\coordinate (b) at (.25,-1);
\node at (.25,-.6) {$\netS$};
\draw[dashed] (.25,-.6) ellipse (10mm and 4mm);
\begin{scope}[every node/.style={rectangle,draw=none}, every path/.style={draw=black}]
	\path (-.5,-.75) -- ++(0,-.5);
	\node at (-.5, -1.45) {$\Zn_1$};
	\path (0,-.75) -- ++(0,-.5);
	\node at (0, -1.45) {$\Zn_2$};
	\path (.5,-.75) -- ++(0,-.5);
	\node at (.5, -1.45) {$\Z_1$};
	\path (1,-.75) - ++(0,-.5);
	\node at (1, -1.45) {$\Y$};
\end{scope}
\end{tikzpicture}}}
\newcommand*{\figSp}{\resizebox*{!}{2em}{
\begin{tikzpicture}[baseline=(b)]
\coordinate (b) at (.25,-1);
\node at (.25,-.6) {$\netS$};
\draw[dashed] (.25,-.6) ellipse (10mm and 4mm);
\begin{scope}[every node/.style={rectangle,draw=none}, every path/.style={draw=black}]
	\path (-.5,-.75) -- ++(0,-.5);
	\node at (-.5, -1.45) {$\Zn_1$};
	\path (0,-.75) -- ++(0,-.5);
	\node at (0, -1.45) {$\Zn_2$};
	\path (.5,-.75) -- ++(0,-.5);
	\node at (.5, -1.45) {$\Z_1$};
	\node (cutY1) at (1.75,-1.25) {$\cut$};
	\node (wY1) at (2.5,-.75) {$\weak$};
	\path (1,-.75) |- node[named edge,above right]{$\Y$} (cutY1);
	\path (wY1) |- node[named edge,above left]{$\Yn$} (cutY1);
\end{scope}
\end{tikzpicture}}}

\section{Graphical reasoning on Conditional Independence}\label{sec:conditional_indepenence}

We demonstrate here how to graphically reason with \BPNs.
As for BNs, a simple and intuitive graphical criterion allows to establish conditional independence between sets of \rvs.
Our contribution is a \emph{direct} and immediate graphical proof of soundness.

\BNs allow to deduce conditional independences  by a graphical criterion known as \emph{d-separation}, first  introduced by Pearl~\cite{Pearl88}.
Here we adapt the  reformulation of d-separation in terms of pruning and disconnectedness given by Darwich in~\cite{DarwicheBook} (Thm.~4.1) as a test which can be decided in time and space that are linear in the size of the DAG. 
A similar graphical criterion is used in the string diagram setting~\cite{FritzK23}.

\subsubsection{Conditional Independence.}

Conditional independence is a fundamental property in probability theory.
It is extremely useful in Bayesian modeling: it \emph{simplifies reasoning} about groups of random variables, and is central to several inference algorithms.

\emph{Intuitively}, two random variables are independent if information about one gives no information about the other.
Independence may be \emph{conditioned} on information coming from a third variable.%
\SLV{}{ Two random variables $X$ and $Y$ are independent conditioned on a third variable $Z$ if for every fixed value of $Z$, information about one of $X$ and $Y$ gives no information about the other.}
\emph{Formally}, let $\Pr$ be a probability distribution over \rvs $\bW$ and $\bX,\bY,\bZ\subseteq \bW$ three disjoint sets of \rvs.
\begin{itemize}
	\item
	$\bX$ and $\bY$ are \textdef{independent}, written $\bX\bot\bY$, if $\Pr (\bX,\bY) = \Pr (\bX)\FProd \Pr (\bY)$.
	\item
	$\bX$ and $\bY$ are \textdef{conditionally independent} given $\bZ$, written $\bX\bot \bY \mid \bZ$, if ${\Pr(\bX,\bY\mid\bZ)=\Pr(\bX\mid\bZ) \FProd \Pr(\bY\mid\bZ)}$.
\end{itemize}

\subsubsection{d-Separation on \BPNs.}
Take $\M:\Gamma$ an arbitrary \emph{positive} \BPN in normal form with respect to both $\Red[\rules]$ and $\Red[\weak]$.
Recall it has an atomic sub-net $\M^\at$ with the same semantic (\cref{fact:canonical_dec}).
The following graphical criterion allows to establish conditional independence in the probability distribution $\sem \M = \sem{\M^\at}$ among \rvs $\Nm{\Gamma}$.

\begin{definition}[Disconnection]\label{def:separation}
For $\M:\Gamma$ a positive \BPN in $\Red[\rules,\weak]$-normal form, let $\{\bX,\bY,\bZ\}$ be a partition of $\Nm{\Gamma}$.
We say that $\bX$ and $\bY$ are \textbf{disconnected} by $\bZ$ if there is no path between $\bbox^X$ and $\bbox^Y$ (for any $X\in \bX$ and $Y\in \bY$) in the graph obtained from $\M^\at$ by removing all the edges labeled by $\bZ$.
\end{definition}

Disconnection implies that $\bX$ and $\bY$ are conditionally independent given $\bZ$.

\begin{theorem}[Soundness of the criterion]\label{prop:separation_independance}
With the same assumptions as in \Cref{def:separation}, if $\bX$ and $\bY$ are disconnected by $\bZ$ then $\bX\bot \bY\mid\bZ$ in the distribution $\sem{\M}$.
\end{theorem}

We   give an immediate, direct proof, by using  a  graphical decomposition together with  compositionality of the semantics.

\subsubsection{Proof.}
Call $\Pr$ the probability distribution associated to $\sem{\M}$.
We want\\${\Pr(\bX,\bY\mid\bZ)=\Pr(\bX\mid\bZ) \Fprod \Pr(\bY\mid\bZ)}$.

\paragraph{Setting the ingredients.}
Let $\R$ be the sub-graph whose nodes are those connected to a box $\bbox^X$ (for $X\in\bX$) in the graph obtained from $\M^\at$ by removing all edges labelled by $\bZ$.  
It is easy to check that $\R$ is a sub-net.
We expand $\M^\at$ into $\pulledout{\M^\at}{\R} = \Cuts{\R,\netS}$ which has the same semantic as $\M^\at$, and so as $\M$.

Please notice that $\Nm{\R} \cap \bY = \emptyset$: otherwise, \cref{lem:jointree} yields a path labelled by $Y\in \bY$ between an edge of $\R$ and $\bbox^Y$ in $\netS$; contradiction as edges between $\R$ and $\netS$ are all labelled by $\bZ$.
Similarly, $\Nm{\netS} \cap \bX = \emptyset$.

\paragraph{Graphical proof.} 
We consider four \BPNs in this proof, see \cref{fig:bpns_dseparation} for an illustration: first $\Cuts{\R,\netS}$, then $\Cuts{\R',\netS}$ the weakening of $\Cuts{\R,\netS}$ by $\bX$, next $\Cuts{\R,\netS'}$ the weakening of $\Cuts{\R,\netS}$ by $\bY$, and finally $\Cuts{\R',\netS'}$ the weakening by both $\bX$ and $\bY$.
In the following, we denote $\R$, $\R'$, $\netS$ and $\netS'$ by the corresponding drawings from \cref{fig:bpns_dseparation} (\eg\ \resizebox*{!}{2em}{
\begin{tikzpicture}[baseline=(b)]
\coordinate (b) at (-4.25,-1);
\node at (-4.25,-.6) {$\R$};
\draw[dashed] (-4.25,-.6) ellipse (10mm and 4mm);
\begin{scope}[every node/.style={rectangle,draw=none}, every path/.style={draw=black}]
	\path (-3.5,-.75) -- ++(0,-.5);
	\node at (-3.5, -1.45) {$\Z_1$};
	\path (-4,-.75) -- ++(0,-.5);
	\node at (-4, -1.45) {$\Z_2$};
	\path (-4.5,-.75) - ++(0,-.5);
	\node at (-4.5, -1.45) {$\Z_2$};
	\node (cutX1) at (-5.75,-1.25) {$\cut$};
	\node (wX1) at (-6.5,-.75) {$\weak$};
	\path (-5,-.75) |- node[named edge,above left]{$\X$} (cutX1);
	\path (wX1) |- node[named edge,above right]{$\Xn$} (cutX1);
\end{scope}
\end{tikzpicture}}\ is $\R'$).


\begin{figure}
\begin{adjustbox}{center=\linewidth,max width=.5\linewidth}
\begin{tikzpicture}
\node at (-2,0) {$\Cuts{\R,\netS}$};
\node at (.25,-.6) {$\netS$};
\draw[dashed] (.25,-.6) ellipse (11mm and 6mm);
\node at (-4.25,-.6) {$\R$};
\draw[dashed] (-4.25,-.6) ellipse (11mm and 6mm);
\begin{scope}[every node/.style={rectangle,draw=none}, every path/.style={draw=black}]
	\node (cut1) at (-2,-1.5) {$\cut$};
	\node (cut2) at (-2,-2) {$\cut$};
	\path (-.5,-.9) |- node[named edge,above left]{$\Zn_1$} (cut1) -| node[named edge,above right]{$\Z_1$} (-3.5,-.9);
	\path (0,-.9) |- node[named edge,above left]{$\Zn_2$} (cut2) -| node[named edge,above right]{$\Z_2$} (-4,-.9);
	\path (.5,-.9) -- ++(0,-1.5);
	\node at (.5, -2.6) {$\Z_1$};
	\path (-4.5,-.9) - ++(0,-1.5);
	\node at (-4.5, -2.6) {$\Z_2$};
	\path (-5,-.9) -- ++(0,-1.5);
	\node at (-5, -2.6) {$\X$};
	\path (1,-.9) - ++(0,-1.5);
	\node at (1, -2.6) {$\Y$};
\end{scope}
\node at (0,-2.85) {\phantom{r}}; 
\end{tikzpicture}
\end{adjustbox}
\vrule
\begin{adjustbox}{center=\linewidth,max width=.5\linewidth}
\begin{tikzpicture}
\node at (-2,0) {$\Cuts{\R',\netS}$};
\node at (.25,-.6) {$\netS$};
\draw[dashed] (.25,-.6) ellipse (11mm and 6mm);
\node at (-4.25,-.6) {$\R$};
\draw[dashed] (-4.25,-.6) ellipse (11mm and 6mm);
\node at (-6,-.25) {$\color{red}\R'$};
\draw[dotted,red] (-4.8,-1.25) ellipse (22mm and 16mm);
\begin{scope}[every node/.style={rectangle,draw=none}, every path/.style={draw=black}]
	\node (cut1) at (-2,-1.5) {$\cut$};
	\node (cut2) at (-2,-2) {$\cut$};
	\path (-.5,-.9) |- node[named edge,above left]{$\Zn_1$} (cut1) -| node[named edge,above right]{$\Z_1$} (-3.5,-.9);
	\path (0,-.9) |- node[named edge,above left]{$\Zn_2$} (cut2) -| node[named edge,above right]{$\Z_2$} (-4,-.9);
	\path (.5,-.9) -- ++(0,-1.5);
	\node at (.5, -2.6) {$\Z_1$};
	\path (-4.5,-.9) - ++(0,-1.5);
	\node at (-4.5, -2.6) {$\Z_2$};
	\node (cutX1) at (-5.75,-2) {$\cut$};
	\node (wX1) at (-6.5,-1) {$\weak$};
	\path (-5,-.9) |- node[named edge,above left]{$\X$} (cutX1);
	\path (wX1) |- node[named edge,above right]{$\Xn$} (cutX1);
	\path (1,-.9) - ++(0,-1.5);
	\node at (1, -2.6) {$\Y$};
\end{scope}
\node at (0,-2.85) {\phantom{r}}; 
\end{tikzpicture}
\end{adjustbox}
\hrule
\begin{adjustbox}{center=\linewidth,max width=.5\linewidth}
\begin{tikzpicture}
\node at (-2,0) {$\Cuts{\R,\netS'}$};
\node at (.25,-.6) {$\netS$};
\draw[dashed] (.25,-.6) ellipse (11mm and 6mm);
\node at (-4.25,-.6) {$\R$};
\draw[dashed] (-4.25,-.6) ellipse (11mm and 6mm);
\node at (2,-.25) {$\color{red}\netS'$};
\draw[dotted,red] (.8,-1.25) ellipse (22mm and 16mm);
\begin{scope}[every node/.style={rectangle,draw=none}, every path/.style={draw=black}]
	\node (cut1) at (-2,-1.5) {$\cut$};
	\node (cut2) at (-2,-2) {$\cut$};
	\path (-.5,-.9) |- node[named edge,above left]{$\Zn_1$} (cut1) -| node[named edge,above right]{$\Z_1$} (-3.5,-.9);
	\path (0,-.9) |- node[named edge,above left]{$\Zn_2$} (cut2) -| node[named edge,above right]{$\Z_2$} (-4,-.9);
	\path (.5,-.9) -- ++(0,-1.5);
	\node at (.5, -2.6) {$\Z_1$};
	\path (-4.5,-.9) - ++(0,-1.5);
	\node at (-4.5, -2.6) {$\Z_2$};
	\path (-5,-.9) -- ++(0,-1.5);
	\node at (-5, -2.6) {$\X$};
	\node (cutY1) at (1.75,-2) {$\cut$};
	\node (wY1) at (2.5,-1) {$\weak$};
	\path (1,-.9) |- node[named edge,above right]{$\Y$} (cutY1);
	\path (wY1) |- node[named edge,above left]{$\Yn$} (cutY1);
\end{scope}
\node at (0,.3) {\phantom{r}}; 
\end{tikzpicture}
\end{adjustbox}
\vrule
\begin{adjustbox}{center=\linewidth,max width=.5\linewidth}
\begin{tikzpicture}
\node at (-2,0) {$\Cuts{\R',\netS'}$};
\node at (.25,-.6) {$\netS$};
\draw[dashed] (.25,-.6) ellipse (11mm and 6mm);
\node at (-4.25,-.6) {$\R$};
\draw[dashed] (-4.25,-.6) ellipse (11mm and 6mm);
\node at (2,-.25) {$\color{red}\netS'$};
\draw[dotted,red] (.8,-1.25) ellipse (22mm and 16mm);
\node at (-6,-.25) {$\color{red}\R'$};
\draw[dotted,red] (-4.8,-1.25) ellipse (22mm and 16mm);
\begin{scope}[every node/.style={rectangle,draw=none}, every path/.style={draw=black}]
	\node (cut1) at (-2,-1.5) {$\cut$};
	\node (cut2) at (-2,-2) {$\cut$};
	\path (-.5,-.9) |- node[named edge,above left]{$\Zn_1$} (cut1) -| node[named edge,above right]{$\Z_1$} (-3.5,-.9);
	\path (0,-.9) |- node[named edge,above left]{$\Zn_2$} (cut2) -| node[named edge,above right]{$\Z_2$} (-4,-.9);
	\path (.5,-.9) -- ++(0,-1.5);
	\node at (.5, -2.6) {$\Z_1$};
	\path (-4.5,-.9) - ++(0,-1.5);
	\node at (-4.5, -2.6) {$\Z_2$};
	\node (cutX1) at (-5.75,-2) {$\cut$};
	\node (wX1) at (-6.5,-1) {$\weak$};
	\path (-5,-.9) |- node[named edge,above left]{$\X$} (cutX1);
	\path (wX1) |- node[named edge,above right]{$\Xn$} (cutX1);
	\node (cutY1) at (1.75,-2) {$\cut$};
	\node (wY1) at (2.5,-1) {$\weak$};
	\path (1,-.9) |- node[named edge,above right]{$\Y$} (cutY1);
	\path (wY1) |- node[named edge,above left]{$\Yn$} (cutY1);
\end{scope}
\node at (0,.3) {\phantom{r}}; 
\end{tikzpicture}
\end{adjustbox}
\caption{\BPNs\ considered in the proof of \cref{prop:separation_independance}}
\label{fig:bpns_dseparation}
\end{figure}

By compositionality of the semantic (\cref{th:compositional}) on each of the four \BPNs:
\begin{flalign*}
\Pr(\bX,\bY,\bZ) &= \sem{\resizebox*{!}{2em}{
\begin{tikzpicture}[baseline=(b)]
\coordinate (b) at (-4.25,-1);
\node at (-4.25,-.6) {$\R$};
\draw[dashed] (-4.25,-.6) ellipse (10mm and 4mm);
\begin{scope}[every node/.style={rectangle,draw=none}, every path/.style={draw=black}]
	\path (-3.5,-.75) -- ++(0,-.5);
	\node at (-3.5, -1.45) {$\Z_1$};
	\path (-4,-.75) -- ++(0,-.5);
	\node at (-4, -1.45) {$\Z_2$};
	\path (-4.5,-.75) - ++(0,-.5);
	\node at (-4.5, -1.45) {$\Z_2$};
	\path (-5,-.75) -- ++(0,-.5);
	\node at (-5, -1.45) {$\X$};
\end{scope}
\end{tikzpicture}}} \FProd \sem{\resizebox*{!}{2em}{
\begin{tikzpicture}[baseline=(b)]
\coordinate (b) at (.25,-1);
\node at (.25,-.6) {$\netS$};
\draw[dashed] (.25,-.6) ellipse (10mm and 4mm);
\begin{scope}[every node/.style={rectangle,draw=none}, every path/.style={draw=black}]
	\path (-.5,-.75) -- ++(0,-.5);
	\node at (-.5, -1.45) {$\Zn_1$};
	\path (0,-.75) -- ++(0,-.5);
	\node at (0, -1.45) {$\Zn_2$};
	\path (.5,-.75) -- ++(0,-.5);
	\node at (.5, -1.45) {$\Z_1$};
	\path (1,-.75) - ++(0,-.5);
	\node at (1, -1.45) {$\Y$};
\end{scope}
\end{tikzpicture}}}
\\
\Pr(\bX,\bZ) &= \sum_\bY \Pr(\bX,\bY,\bZ) = \sem{\resizebox*{!}{2em}{
\begin{tikzpicture}[baseline=(b)]
\coordinate (b) at (-4.25,-1);
\node at (-4.25,-.6) {$\R$};
\draw[dashed] (-4.25,-.6) ellipse (10mm and 4mm);
\begin{scope}[every node/.style={rectangle,draw=none}, every path/.style={draw=black}]
	\path (-3.5,-.75) -- ++(0,-.5);
	\node at (-3.5, -1.45) {$\Z_1$};
	\path (-4,-.75) -- ++(0,-.5);
	\node at (-4, -1.45) {$\Z_2$};
	\path (-4.5,-.75) - ++(0,-.5);
	\node at (-4.5, -1.45) {$\Z_2$};
	\path (-5,-.75) -- ++(0,-.5);
	\node at (-5, -1.45) {$\X$};
\end{scope}
\end{tikzpicture}}} \FProd \sem{\resizebox*{!}{2em}{
\begin{tikzpicture}[baseline=(b)]
\coordinate (b) at (.25,-1);
\node at (.25,-.6) {$\netS$};
\draw[dashed] (.25,-.6) ellipse (10mm and 4mm);
\begin{scope}[every node/.style={rectangle,draw=none}, every path/.style={draw=black}]
	\path (-.5,-.75) -- ++(0,-.5);
	\node at (-.5, -1.45) {$\Zn_1$};
	\path (0,-.75) -- ++(0,-.5);
	\node at (0, -1.45) {$\Zn_2$};
	\path (.5,-.75) -- ++(0,-.5);
	\node at (.5, -1.45) {$\Z_1$};
	\node (cutY1) at (1.75,-1.25) {$\cut$};
	\node (wY1) at (2.5,-.75) {$\weak$};
	\path (1,-.75) |- node[named edge,above right]{$\Y$} (cutY1);
	\path (wY1) |- node[named edge,above left]{$\Yn$} (cutY1);
\end{scope}
\end{tikzpicture}}}
\\
\Pr(\bY,\bZ) &= \sum_\bX \Pr(\bX,\bY,\bZ) = \sem{\resizebox*{!}{2em}{
\begin{tikzpicture}[baseline=(b)]
\coordinate (b) at (-4.25,-1);
\node at (-4.25,-.6) {$\R$};
\draw[dashed] (-4.25,-.6) ellipse (10mm and 4mm);
\begin{scope}[every node/.style={rectangle,draw=none}, every path/.style={draw=black}]
	\path (-3.5,-.75) -- ++(0,-.5);
	\node at (-3.5, -1.45) {$\Z_1$};
	\path (-4,-.75) -- ++(0,-.5);
	\node at (-4, -1.45) {$\Z_2$};
	\path (-4.5,-.75) - ++(0,-.5);
	\node at (-4.5, -1.45) {$\Z_2$};
	\node (cutX1) at (-5.75,-1.25) {$\cut$};
	\node (wX1) at (-6.5,-.75) {$\weak$};
	\path (-5,-.75) |- node[named edge,above left]{$\X$} (cutX1);
	\path (wX1) |- node[named edge,above right]{$\Xn$} (cutX1);
\end{scope}
\end{tikzpicture}}} \FProd \sem{\resizebox*{!}{2em}{
\begin{tikzpicture}[baseline=(b)]
\coordinate (b) at (.25,-1);
\node at (.25,-.6) {$\netS$};
\draw[dashed] (.25,-.6) ellipse (10mm and 4mm);
\begin{scope}[every node/.style={rectangle,draw=none}, every path/.style={draw=black}]
	\path (-.5,-.75) -- ++(0,-.5);
	\node at (-.5, -1.45) {$\Zn_1$};
	\path (0,-.75) -- ++(0,-.5);
	\node at (0, -1.45) {$\Zn_2$};
	\path (.5,-.75) -- ++(0,-.5);
	\node at (.5, -1.45) {$\Z_1$};
	\path (1,-.75) - ++(0,-.5);
	\node at (1, -1.45) {$\Y$};
\end{scope}
\end{tikzpicture}}}
\\
\Pr(\bZ) &= \sum_{\bX,\bY} \Pr(\bX,\bY,\bZ) = \sem{\resizebox*{!}{2em}{
\begin{tikzpicture}[baseline=(b)]
\coordinate (b) at (-4.25,-1);
\node at (-4.25,-.6) {$\R$};
\draw[dashed] (-4.25,-.6) ellipse (10mm and 4mm);
\begin{scope}[every node/.style={rectangle,draw=none}, every path/.style={draw=black}]
	\path (-3.5,-.75) -- ++(0,-.5);
	\node at (-3.5, -1.45) {$\Z_1$};
	\path (-4,-.75) -- ++(0,-.5);
	\node at (-4, -1.45) {$\Z_2$};
	\path (-4.5,-.75) - ++(0,-.5);
	\node at (-4.5, -1.45) {$\Z_2$};
	\node (cutX1) at (-5.75,-1.25) {$\cut$};
	\node (wX1) at (-6.5,-.75) {$\weak$};
	\path (-5,-.75) |- node[named edge,above left]{$\X$} (cutX1);
	\path (wX1) |- node[named edge,above right]{$\Xn$} (cutX1);
\end{scope}
\end{tikzpicture}}} \FProd \sem{\resizebox*{!}{2em}{
\begin{tikzpicture}[baseline=(b)]
\coordinate (b) at (.25,-1);
\node at (.25,-.6) {$\netS$};
\draw[dashed] (.25,-.6) ellipse (10mm and 4mm);
\begin{scope}[every node/.style={rectangle,draw=none}, every path/.style={draw=black}]
	\path (-.5,-.75) -- ++(0,-.5);
	\node at (-.5, -1.45) {$\Zn_1$};
	\path (0,-.75) -- ++(0,-.5);
	\node at (0, -1.45) {$\Zn_2$};
	\path (.5,-.75) -- ++(0,-.5);
	\node at (.5, -1.45) {$\Z_1$};
	\node (cutY1) at (1.75,-1.25) {$\cut$};
	\node (wY1) at (2.5,-.75) {$\weak$};
	\path (1,-.75) |- node[named edge,above right]{$\Y$} (cutY1);
	\path (wY1) |- node[named edge,above left]{$\Yn$} (cutY1);
\end{scope}
\end{tikzpicture}}}
\end{flalign*}
(because summing out corresponds to weakening, this follows from compositionality).

Fix values $\bx$, $\by$ and $\bz$.
Call $\bz_r$ the (values of the) variables of $\bz$ appearing in the conclusions of $\R$, and $\bz_s$ those appearing in the conclusions of $\netS$.
Then:
{\allowdisplaybreaks
\begin{flalign*}
\Pr(\bx\mid\bz) \Pr(\by\mid\bz)
&=
\frac{\Pr(\bx, \bz)}{\Pr(\bz)} \frac{\Pr(\by, \bz)}{\Pr(\bz)}
\\
&=
\frac{\sem{\resizebox*{!}{2em}{
\begin{tikzpicture}[baseline=(b)]
\coordinate (b) at (-4.25,-1);
\node at (-4.25,-.6) {$\R$};
\draw[dashed] (-4.25,-.6) ellipse (10mm and 4mm);
\begin{scope}[every node/.style={rectangle,draw=none}, every path/.style={draw=black}]
	\path (-3.5,-.75) -- ++(0,-.5);
	\node at (-3.5, -1.45) {$\Z_1$};
	\path (-4,-.75) -- ++(0,-.5);
	\node at (-4, -1.45) {$\Z_2$};
	\path (-4.5,-.75) - ++(0,-.5);
	\node at (-4.5, -1.45) {$\Z_2$};
	\path (-5,-.75) -- ++(0,-.5);
	\node at (-5, -1.45) {$\X$};
\end{scope}
\end{tikzpicture}}}_{\bx,\bz_r} \sem{\resizebox*{!}{2em}{
\begin{tikzpicture}[baseline=(b)]
\coordinate (b) at (.25,-1);
\node at (.25,-.6) {$\netS$};
\draw[dashed] (.25,-.6) ellipse (10mm and 4mm);
\begin{scope}[every node/.style={rectangle,draw=none}, every path/.style={draw=black}]
	\path (-.5,-.75) -- ++(0,-.5);
	\node at (-.5, -1.45) {$\Zn_1$};
	\path (0,-.75) -- ++(0,-.5);
	\node at (0, -1.45) {$\Zn_2$};
	\path (.5,-.75) -- ++(0,-.5);
	\node at (.5, -1.45) {$\Z_1$};
	\node (cutY1) at (1.75,-1.25) {$\cut$};
	\node (wY1) at (2.5,-.75) {$\weak$};
	\path (1,-.75) |- node[named edge,above right]{$\Y$} (cutY1);
	\path (wY1) |- node[named edge,above left]{$\Yn$} (cutY1);
\end{scope}
\end{tikzpicture}}}_{\bz_s} \sem{\resizebox*{!}{2em}{
\begin{tikzpicture}[baseline=(b)]
\coordinate (b) at (-4.25,-1);
\node at (-4.25,-.6) {$\R$};
\draw[dashed] (-4.25,-.6) ellipse (10mm and 4mm);
\begin{scope}[every node/.style={rectangle,draw=none}, every path/.style={draw=black}]
	\path (-3.5,-.75) -- ++(0,-.5);
	\node at (-3.5, -1.45) {$\Z_1$};
	\path (-4,-.75) -- ++(0,-.5);
	\node at (-4, -1.45) {$\Z_2$};
	\path (-4.5,-.75) - ++(0,-.5);
	\node at (-4.5, -1.45) {$\Z_2$};
	\node (cutX1) at (-5.75,-1.25) {$\cut$};
	\node (wX1) at (-6.5,-.75) {$\weak$};
	\path (-5,-.75) |- node[named edge,above left]{$\X$} (cutX1);
	\path (wX1) |- node[named edge,above right]{$\Xn$} (cutX1);
\end{scope}
\end{tikzpicture}}}_{\bz_r} \sem{\resizebox*{!}{2em}{
\begin{tikzpicture}[baseline=(b)]
\coordinate (b) at (.25,-1);
\node at (.25,-.6) {$\netS$};
\draw[dashed] (.25,-.6) ellipse (10mm and 4mm);
\begin{scope}[every node/.style={rectangle,draw=none}, every path/.style={draw=black}]
	\path (-.5,-.75) -- ++(0,-.5);
	\node at (-.5, -1.45) {$\Zn_1$};
	\path (0,-.75) -- ++(0,-.5);
	\node at (0, -1.45) {$\Zn_2$};
	\path (.5,-.75) -- ++(0,-.5);
	\node at (.5, -1.45) {$\Z_1$};
	\path (1,-.75) - ++(0,-.5);
	\node at (1, -1.45) {$\Y$};
\end{scope}
\end{tikzpicture}}}_{\by,\bz_s}}{\sem{\resizebox*{!}{2em}{
\begin{tikzpicture}[baseline=(b)]
\coordinate (b) at (-4.25,-1);
\node at (-4.25,-.6) {$\R$};
\draw[dashed] (-4.25,-.6) ellipse (10mm and 4mm);
\begin{scope}[every node/.style={rectangle,draw=none}, every path/.style={draw=black}]
	\path (-3.5,-.75) -- ++(0,-.5);
	\node at (-3.5, -1.45) {$\Z_1$};
	\path (-4,-.75) -- ++(0,-.5);
	\node at (-4, -1.45) {$\Z_2$};
	\path (-4.5,-.75) - ++(0,-.5);
	\node at (-4.5, -1.45) {$\Z_2$};
	\node (cutX1) at (-5.75,-1.25) {$\cut$};
	\node (wX1) at (-6.5,-.75) {$\weak$};
	\path (-5,-.75) |- node[named edge,above left]{$\X$} (cutX1);
	\path (wX1) |- node[named edge,above right]{$\Xn$} (cutX1);
\end{scope}
\end{tikzpicture}}}_{\bz_r} \sem{\resizebox*{!}{2em}{
\begin{tikzpicture}[baseline=(b)]
\coordinate (b) at (.25,-1);
\node at (.25,-.6) {$\netS$};
\draw[dashed] (.25,-.6) ellipse (10mm and 4mm);
\begin{scope}[every node/.style={rectangle,draw=none}, every path/.style={draw=black}]
	\path (-.5,-.75) -- ++(0,-.5);
	\node at (-.5, -1.45) {$\Zn_1$};
	\path (0,-.75) -- ++(0,-.5);
	\node at (0, -1.45) {$\Zn_2$};
	\path (.5,-.75) -- ++(0,-.5);
	\node at (.5, -1.45) {$\Z_1$};
	\node (cutY1) at (1.75,-1.25) {$\cut$};
	\node (wY1) at (2.5,-.75) {$\weak$};
	\path (1,-.75) |- node[named edge,above right]{$\Y$} (cutY1);
	\path (wY1) |- node[named edge,above left]{$\Yn$} (cutY1);
\end{scope}
\end{tikzpicture}}}_{\bz_s} \sem{\resizebox*{!}{2em}{
\begin{tikzpicture}[baseline=(b)]
\coordinate (b) at (-4.25,-1);
\node at (-4.25,-.6) {$\R$};
\draw[dashed] (-4.25,-.6) ellipse (10mm and 4mm);
\begin{scope}[every node/.style={rectangle,draw=none}, every path/.style={draw=black}]
	\path (-3.5,-.75) -- ++(0,-.5);
	\node at (-3.5, -1.45) {$\Z_1$};
	\path (-4,-.75) -- ++(0,-.5);
	\node at (-4, -1.45) {$\Z_2$};
	\path (-4.5,-.75) - ++(0,-.5);
	\node at (-4.5, -1.45) {$\Z_2$};
	\node (cutX1) at (-5.75,-1.25) {$\cut$};
	\node (wX1) at (-6.5,-.75) {$\weak$};
	\path (-5,-.75) |- node[named edge,above left]{$\X$} (cutX1);
	\path (wX1) |- node[named edge,above right]{$\Xn$} (cutX1);
\end{scope}
\end{tikzpicture}}}_{\bz_r} \sem{\resizebox*{!}{2em}{
\begin{tikzpicture}[baseline=(b)]
\coordinate (b) at (.25,-1);
\node at (.25,-.6) {$\netS$};
\draw[dashed] (.25,-.6) ellipse (10mm and 4mm);
\begin{scope}[every node/.style={rectangle,draw=none}, every path/.style={draw=black}]
	\path (-.5,-.75) -- ++(0,-.5);
	\node at (-.5, -1.45) {$\Zn_1$};
	\path (0,-.75) -- ++(0,-.5);
	\node at (0, -1.45) {$\Zn_2$};
	\path (.5,-.75) -- ++(0,-.5);
	\node at (.5, -1.45) {$\Z_1$};
	\node (cutY1) at (1.75,-1.25) {$\cut$};
	\node (wY1) at (2.5,-.75) {$\weak$};
	\path (1,-.75) |- node[named edge,above right]{$\Y$} (cutY1);
	\path (wY1) |- node[named edge,above left]{$\Yn$} (cutY1);
\end{scope}
\end{tikzpicture}}}_{\bz_s}}
\\
&=
\frac{\sem{\resizebox*{!}{2em}{
\begin{tikzpicture}[baseline=(b)]
\coordinate (b) at (-4.25,-1);
\node at (-4.25,-.6) {$\R$};
\draw[dashed] (-4.25,-.6) ellipse (10mm and 4mm);
\begin{scope}[every node/.style={rectangle,draw=none}, every path/.style={draw=black}]
	\path (-3.5,-.75) -- ++(0,-.5);
	\node at (-3.5, -1.45) {$\Z_1$};
	\path (-4,-.75) -- ++(0,-.5);
	\node at (-4, -1.45) {$\Z_2$};
	\path (-4.5,-.75) - ++(0,-.5);
	\node at (-4.5, -1.45) {$\Z_2$};
	\path (-5,-.75) -- ++(0,-.5);
	\node at (-5, -1.45) {$\X$};
\end{scope}
\end{tikzpicture}}}_{\bx,\bz_r} \sem{\resizebox*{!}{2em}{
\begin{tikzpicture}[baseline=(b)]
\coordinate (b) at (.25,-1);
\node at (.25,-.6) {$\netS$};
\draw[dashed] (.25,-.6) ellipse (10mm and 4mm);
\begin{scope}[every node/.style={rectangle,draw=none}, every path/.style={draw=black}]
	\path (-.5,-.75) -- ++(0,-.5);
	\node at (-.5, -1.45) {$\Zn_1$};
	\path (0,-.75) -- ++(0,-.5);
	\node at (0, -1.45) {$\Zn_2$};
	\path (.5,-.75) -- ++(0,-.5);
	\node at (.5, -1.45) {$\Z_1$};
	\path (1,-.75) - ++(0,-.5);
	\node at (1, -1.45) {$\Y$};
\end{scope}
\end{tikzpicture}}}_{\by,\bz_s}}{\sem{\resizebox*{!}{2em}{
\begin{tikzpicture}[baseline=(b)]
\coordinate (b) at (-4.25,-1);
\node at (-4.25,-.6) {$\R$};
\draw[dashed] (-4.25,-.6) ellipse (10mm and 4mm);
\begin{scope}[every node/.style={rectangle,draw=none}, every path/.style={draw=black}]
	\path (-3.5,-.75) -- ++(0,-.5);
	\node at (-3.5, -1.45) {$\Z_1$};
	\path (-4,-.75) -- ++(0,-.5);
	\node at (-4, -1.45) {$\Z_2$};
	\path (-4.5,-.75) - ++(0,-.5);
	\node at (-4.5, -1.45) {$\Z_2$};
	\node (cutX1) at (-5.75,-1.25) {$\cut$};
	\node (wX1) at (-6.5,-.75) {$\weak$};
	\path (-5,-.75) |- node[named edge,above left]{$\X$} (cutX1);
	\path (wX1) |- node[named edge,above right]{$\Xn$} (cutX1);
\end{scope}
\end{tikzpicture}}}_{\bz_r} \sem{\resizebox*{!}{2em}{
\begin{tikzpicture}[baseline=(b)]
\coordinate (b) at (.25,-1);
\node at (.25,-.6) {$\netS$};
\draw[dashed] (.25,-.6) ellipse (10mm and 4mm);
\begin{scope}[every node/.style={rectangle,draw=none}, every path/.style={draw=black}]
	\path (-.5,-.75) -- ++(0,-.5);
	\node at (-.5, -1.45) {$\Zn_1$};
	\path (0,-.75) -- ++(0,-.5);
	\node at (0, -1.45) {$\Zn_2$};
	\path (.5,-.75) -- ++(0,-.5);
	\node at (.5, -1.45) {$\Z_1$};
	\node (cutY1) at (1.75,-1.25) {$\cut$};
	\node (wY1) at (2.5,-.75) {$\weak$};
	\path (1,-.75) |- node[named edge,above right]{$\Y$} (cutY1);
	\path (wY1) |- node[named edge,above left]{$\Yn$} (cutY1);
\end{scope}
\end{tikzpicture}}}_{\bz_s}}
\\
&=
\frac{\Pr(\bx,\by,\bz)}{\Pr(\bz)}
=
\Pr(\bx,\by\mid\bz)
\end{flalign*}
}

Hence $\Pr(\bX,\bY\mid\bZ)=\Pr(\bX\mid\bZ) \Fprod \Pr(\bY\mid\bZ)$ as wished.


%
%

\section{Computing the Semantics, efficiently}
\label{sec:computing_efficiently}




Let us consider the cost of computing the semantics of a \BPN\ $\R$ as defined in \Cref{def:sem}.
Denote by $m$ is the number of boxes in $\R$ and $n$ the cardinality of $\Nm{\Boxes{\R}}$.
Recalling \Cref{rem:factors_cost}, computing $\sem{\R}$ according to the formula of \Cref{def:sem} requires to compute and store $\BigO{m \cdot 2^{n}}$ values.
However, \emph{we can do better} if we are able to \emph{factorize} a \BPN in the composition of smaller nets, whose interpretations have a smaller cost.

In proof-theory, the natural way to factorize  a proof in smaller components is to factorize  it in sub-proofs which are  composed together via cuts. 
\begin{example}[Roadmap]\label{ex:roadmap}
Take the \BPN $\R_D$ in \Cref{fig:ABCDE3}.
To interpret it via \Cref{def:sem}, we need to compute $2^5$ values, and then to marginalize.
However, simply by $\ax$-expansion, we can rewrite $\R_D$ into $\R_D'$ of \Cref{fig:ABCDE4}, which (because of compositionality) we are able to  interpret at a smaller cost, as we show in this section. Notice that $\phisem{\R'_D} =\phisem{\R_D}$.
\end{example}

\SLV{}{\pink{
In \Cref{sec:MP} we will  show\todoc{If time allows} that there is a tight correspondence between the decomposition of  a proof-net into a cut-net (introduced below), 
and the  well-known transformation of Bayesian Networks into clique trees.
The correspondence turns out to be so tight, that the cost of computing the interpretation of a cut-net is the same as the cost of inference on  the corresponding clique tree.}}

\subsection{Cut-nets}\label{sec:factorization}
In the literature of linear logic, the decomposition of a proof-net in sub-nets linked together by $\cut$-nodes is called a \textit{cut-net}~\cite{locus}.

\begin{definition}[Cut-net]\label{def:cut-net}
	We write $\R = \CutS{\R_1, \dots, \R_n}$ to denote a \pn $\R$ for which is given a \emph{partition} into sub-nets $\R_1, \dots, \R_n$ and $\cut$-nodes which  \emph{separate} the sub-nets (the  edges of each such $\cut$ are conclusions of  distinct sub-nets).

	$\CutS{\R_1, \dots, \R_n}$ is a \textbf{\cutnet} if the skeleton  graph which has nodes $\{\R_1,...,\R_n\}$ and an edge $(\R_i,\R_j)$ exactly when there is at least one cut between $\R_i$ and $\R_j$ is a \emph{tree}.
\end{definition}

\begin{example}[Cut-net]\label{ex:cutnet}
The \BPN in \Cref{fig:ABCDE4} can be partitioned as a \cutnet with 3 components (and 4 cuts).
The \BPN $\R_D$ in \Cref{fig:ABCDE3} is a (trivial) cut-net which has a single component.
\end{example}

A cut-net is a way to organize a \pn into a \emph{tree-like} structure. Please notice that the choice of the sub-nets is fully arbitrary: each sub-net may itself be a cut-net.
When we need an explicit tree-structure, we use the following notation to indicate which sub-net is chosen as the root, and call the cut-net \emph{rooted}.
The choice of a root is arbitrary: every sub-net can play such a role.



\begin{notation}[Rooted cut-net]
We write $\N= \R(\M_1, \dots, \M_h)$ to denote a factorization of $\N$ as a cut-net where the sub-net $\R$, the root, is connected to each $\M_i$,  these $\M_i$ being themselves cut-nets (with a root).
\end{notation}

\begin{figure}[t]
\begin{minipage}[c]{0.48\linewidth}
	\includegraphics[page=3,width=\linewidth]{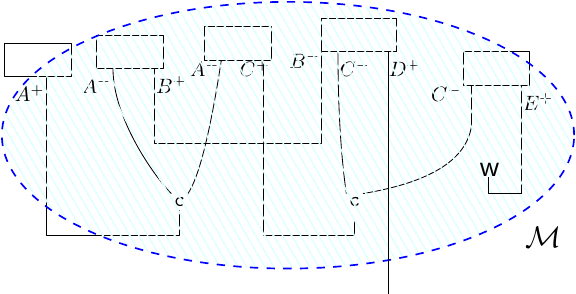}\vspace*{-5mm}
	\caption{$\R_D$}
	\label{fig:ABCDE3}
\end{minipage}
\hfill
\begin{minipage}[c]{0.48\linewidth}
	\includegraphics[page=4,width=\linewidth]{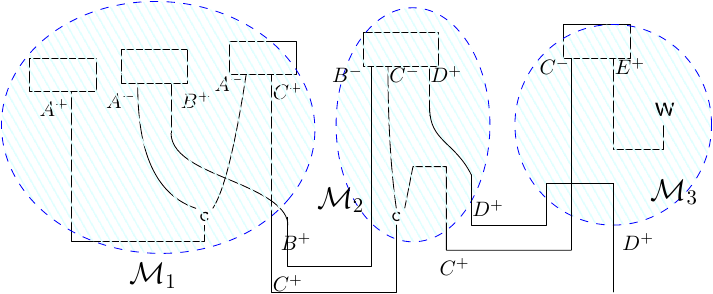}
	\caption{$\R_D'$. \quad Notice $\sem {\R_D}= \sem{\R_D'}$}
	\label{fig:ABCDE4}
\end{minipage}
\end{figure}

\begin{example}[Rooted cut-net]\label{ex:fact}
	Consider the  cut-net $\R_D' = \Cuts{\M_1,\M_2, \M_3}$ in 
	\Cref{fig:ABCDE4}. By choosing $\M_1$ (\resp\ $\M_2, \M_3$) as root, we can write this cut-net as
	\begin{center}
	$\R_D'=\M_1 (\M_2(\M_3))$ \quad or \quad $\R_D'=\M_2(\M_1,\M_3)$ \quad or \quad $\R_D'=\M_3(\M_2(\M_1))$
	\end{center}
\end{example}



\subsection{Inference by Interpretation, efficiently}\label{sec:turbo}

Thanks to compositionality, we can inductively interpret a rooted cut-net by following its tree-shape.

\begin{theorem}[Inductive  interpretation]\label{cor:turbo}\label{thm:turbo}
	Let $\R(\N_1, \dots, \N_n):\Delta$ be a rooted cut-net. Noting $\bX=\bigcup_{\bbox\in \Boxes{\pi}} \Nm{\bbox}$, we have:
	\[\phisem{\R(\N_1, \dots, \N_n):\Delta} = \sum_{\bX - \Nm{\Delta}} \left( \sem {\R} \odot \BigFProd_{i:1}^n\phisem{\N_i} \right)\]
\end{theorem}

\begin{example}[Factorized Interpretation]\label{ex:fact_2}  
	Let us compute the interpretation of the cut-net $\sem{\R_D'}$ in \Cref{fig:ABCDE4} rooted on $\M_3$ as $\M_3(\M_2(\M_1))$.
	\begin{center}
		{\small 		$ 	\begin{array}{c|c|c|c}
				\text{cut-net} & \text{computations} &\text{resulting factor} &\Names  \\
				\hline
				\M_1	& \sum_A \phi^A\FProd \phi^B\FProd \phi^C & \phi^1(B,C)& ABC  \\
				\M_2(\M_1)	&  \sum_B \phi^D\FProd \phi^1(B,C) & \phi^2(C,D) & BCD\\
				\M_3(\M_2(\M_1))	& \sum_{C,E}\phi^E\FProd \phi^2(C,D) & \phi^3(D)& CDE
			\end{array} $}
	\end{center}
	Computing this way the semantics of $\R_D'$, the size of the largest intermediary factor that is computed is $2^3$.
	This is to be contrasted with the size of $2^5$ in \Cref{ex:roadmap}.
\end{example}

\subsubsection{Cost Analysis.}

The crucial parameter to determine the cost of the interpretation is the
number of names in each components, which we call width (in analogy with  similar notions in BN's.).
Given a cut-net $\N=\Cuts{\R_1, \dots, \R_n}$,
we define  $\Width{\N}$ as the largest cardinality $|\Nm{\R_i}|$  of any of its sub-nets, minus one.

\begin{example}The cut-net  in \Cref{fig:ABCDE3} has width $4$,  the cut-net in \Cref{fig:ABCDE4} has width $2$. 
\end{example}


\begin{theorem}[Cost of interpreting a cut-net]\label{prop:cost} Let $\N$ be a positive cut-net and $n = |\Nm{\N}|$, the number of names appearing in $\N$.
	The time and space cost of inductively computing $\phisem \N$ is $\BigO{n \cdot 2^ {\Width\N}}$.
\end{theorem}
\SLV{}{
	\begin{proof}
		First, observe that  $Ax{(\R)}\cup\Wirings{\R}$ are the components of the cut-net $\R$, and so the nodes of the correction graph $\tree_\R$. So 
		$m_\R$  counts the number of edges in $\tree_{\R}$ (the number of nodes minus one).
		
		Let $ \Width\R=w $.
		Assume $\R=\D{(\N_1:\Gamma_1, \dots, \N_h:\Gamma_h)}$.	Each   $\phi_i =\sem {\N_i}$ is a factor over $\At{\Gamma_i}\subseteq \At{\D}$, so  that  $\size{\At{\Gamma_i}} \leq  \size{\At{\D}}\leq (w +1)$.   The cost of computing 
		$\BigFProd_{i:1}^h\phi_i$ is  $\BigO{h\cdot \Exp{w+1} }$. Since the cost of the projection is similar, the cost of computing $\sem{\D(\N_1, \dots, \N_h)}$ is  $\BigO{h \cdot \Exp{w}}$ \emph{plus} the cost  of  computing 
		$\phi_1=\sem{\N_1}, \dots \phi_h=\sem{\N_h}$.
		
		
		By \ih, for each $i$, the cost of computing   $\sem{\N_i}$ is $\BigO{m_{\N_i} \cdot \Exp {w}}$, so 
		$\sem{\D({\N_1}, \dots, {\N_h})} = \BigO{(h + \sum_{i:1}^h m_{\N_i}) \cdot \Exp{w} }  $ 
		
		Recall that   $m_{\R}$ is the number of edges in the tree  $\tree_{\R}$, which  has root $\D$, and a subtree $\tree_{\N_i}$ (with $m_{\N_i}$ edges) for each $\N_i$.
		So the number of edges in $\tree_{\R}$ is  $m_{\R} = h+  \sum_{i:1}^h m_{\N_i}$. We conclude that  
		the total cost for computing $\sem{\R}$ 	is	\begin{center}
			$ \BigO{m_{\R} \cdot \Exp {\Width\R}} $.
		\end{center}
	\end{proof}
}

\subsubsection{Factorizing a \BPN, and factorized algorithms.}

We can create a cut-net starting from any \pn by  repeated applications of  splitting  (\cref{def:decomposition}).
From there, we can mimic in \BPN  exact inference algorithms from \BNs, by building a cut-net whose interpretation using \cref{cor:turbo} behaves as the wished computation. 
For example, the \pn on \cref{fig:ABCDE3} can be expanded as the one on \cref{fig:ABCDE4}, that has a smaller cost, by doing the splitting induced by $\M_1$ and then the one induced by $\M_2$ -- and this corresponds to the \emph{variable elimination} algorithm~\cite{DarwicheBook}~(Chapter~6) with the elimination order $(A,B,C)$.
In \cref{sec:ve_in_bpn} is sketched how to mimic variable elimination on \BPNs.

\subsection{Cut-Nets and Typing}
\label{sec:typing_cut_net}

In this paper, we extensively used splitting (\Cref{def:decomposition}) to decompose a \BPN $\R$ into a cut-net $\CutS{\R_1, \dots, \R_n}$ made of several sub-nets $\R_i$ joined by (possibly many) $\cut$-nodes (recall \Cref{def:cut-net}).
This  graphical decomposition has \textit{a priori} no correspondance in \pts of sequent calculus.
Indeed, while each of the $\R_i$ is the image of a \pt $\pi_i$, there is no sequent calculus rule allowing to do multiple cuts at once in sequent calculus, whereas it is quite natural in \pns.
In other words, while $\R$ is for sure the image of a \pt $\pi$, the inductive shape of \pts bears no similarity with the decomposition: there may be no $\pi$ whose image is $\R$ and which contains all $\pi_i$ as sub-proofs.

We prove here that the decomposition of $\R$ as $\CutS{\R_1, \dots, \R_n}$ can be translated back in sequent calculus, \emph{up to some rewriting on $\R$}: $\R$ is the image of a \pt $\pi$ containing as sub-trees the $\pi_i$ whose images are the sub-nets $\R_i$.
This gives a proof-theoretic counterpart to the diagrammatic reasoning.
The crucial  technical result allowing this translation is the following: by $\otimes/\parr$-expansion (the reverse of the $\otimes/\parr$-rule in \cref{fig:rewriting}), we transform the  cut-net $\R$  in an equivalent  one, where there is at most \emph{one} $\cut$-node between any  two  sub-nets. This is obtained by a typing procedure which progressively  turns  $n$ $\cut$-nodes  on the formulas $A_1, \dots, A_n$  into a single $\cut$ on a large  formula built over 
$A_1, \dots, A_n$.
As an example, consider $\R_D'$ on \cref{fig:ABCDE4}: it has two cuts between $\M_1$ and $\M_2$, on formulas $B^+$ and $C^+$, and two cuts between $\M_2$ and $\M_3$, on formulas $C^+$ and $D^+$.
Performing two $\otimes/\parr$-expansions, we can obtain a \emph{correct} cut-net with one cut between $\M_1$ and $\M_2$, on the formula $B^+ \otimes C^+$, and one cut between $\M_2$ and $\M_3$, on $C^+ \parr D^+$.

Given a cut-net with at most one cut between any two sub-nets, its translation as a \pt is immediate, because the tree-structure of the cut-net is explicit: it suffices to take \pts corresponding to its sub-nets, and compose them by $\cut$-rules corresponding to each $\cut$-node.
Going back on our example, here is the start of a \pt corresponding to the expansion of $\R_D'$, with $\pi_i$ a \pt corresponding to $\M_i$ (with additional $\otimes$ and $\parr$ added during the expansion):
\begin{center}
{\footnotesize
$\infer[\cut(C^+ \parr D^+)]{\vdash D^+}{
 \infer[\cut(B^+\otimes C^+)]{\vdash C^+ \parr D^+}{{\infer{\vdash B^+\otimes C^+}{\pi_1}}&   \infer{\vdash B^-\parr C^-,C^+ \parr D^+}{\pi_2}}
 &   \infer{\vdash C^-\otimes D^-, D^+}{\pi_3}
 }$
 }
\end{center}

\begin{theorem}
\label{th:typingcutnet}
Given a cut-net $\R = \CutS{\R_1, \dots, \R_n}$, there is a sequence of $\otimes/\parr$-expansions (the reverse of the $\otimes/\parr$-rule in \cref{fig:rewriting}) on the cuts between the pairs $(\R_i,\R_j)$ such that the result is a cut-net $\R' = \CutS{\R_1', \dots, \R_n'}$ with \textbf{at most one cut} between each pair $(\R_i', \R_j')$.
\end{theorem}

The  proof (in \cref{sec:proof_typing}) is constructive, and yields a linear algorithm for turning a cut-net into one with at most one cut between each pair of sub-nets.
Between each pair $(\R_i', \R_j')$, the type of the obtained unique cut (\ie\ the formula labelling it) is obtained from the types of the starting cuts between $R_i$ and $R_j$ by means of $\parr$ and $\otimes$.

\bibliographystyle{splncs04}
\bibliography{biblioB,biblioBN}

\newpage
\appendix


\section*{APPENDIX}

\section{Multiplicative Linear Logic with Boxes}
\label{app:sequent}

\condinc{}{
We point out some properties which we use in the proofs.
\begin{remark}\label{rk:subnet}
	Notice that a sub-graph of a corect typed graph is correct. 
\end{remark}

\begin{remark}\label{rk:normal_atomic}
	If $\ProofN$ is a normal \pn with atomic conclusions, then $\ProofN$ is atomic.
\end{remark}

\begin{property}[normal proof-nets]\label{lem:normal_net}
	Let $\R:\Delta$ be a  \emph{normal} \pn. Then 
	(1.) the premises of each $\cut$-node are atomic and
	(2.) the positive premise of each cut-node is conclusion of a $\Ax$-node. 
\end{property}

	\begin{lemma*}[\ref{lem:pol_correct} Polarized correctness] For $\M$  an atomic  \emph{raw} module, we denote by  $\pol{\M}$  
		the  graph which has the same nodes and edges  as $\M$,  but where the edges are directed 
		\SLV{}{according to their polarity:} downward if positive,  upwards if negative. 	
		The following are equivalent: (1.)
		$\M$ is acyclic correct, (2.)
		$\pol \M$   is a DAG.
	\end{lemma*}		
	
	{\begin{proof}We only prove $2. \Rightarrow 1.$ (the other direction is immediate). Recall that a DAG where each node has at most an outcoming edge is a tree. By inspecting the grammar of nodes, we observe that, w.r.t. the \emph{polarized orientation},
			only $\parr$ and $@$ nodes may have more than one outcoming    edges. Hence each switching graph is a tree. 
	\end{proof}}
}

\Cref{fig:sequentialization} gives the rules to generate the sequent calculus \pts for \MLLAx. It also gives, for each \pt  
 $ \pi $ its  image $\N$ as a proof-net. Reading the figure right to left, we have a   sequentialization of the proof-net $\N$.
\begin{figure}
	\centering
	\includegraphics[page=5,width=1\linewidth]{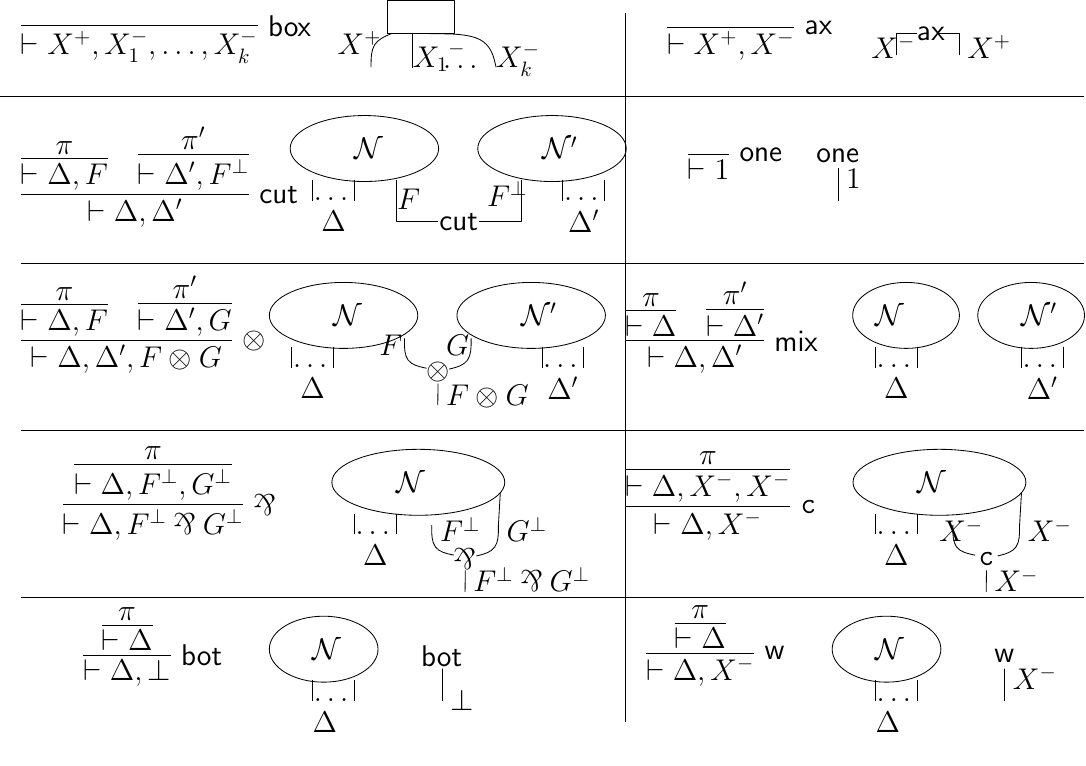}
	\caption{Sequent Calculus, Proof-Nets, and Sequentialization}
	\label{fig:sequentialization}
\end{figure}

\condinc{}{
$ \infer[\mathsf{box}]{\vdash\X,\Xn_1,\dots,\Xn_k}{} $

$ \infer[\ax]{\vdash \X,\Xn}{} $

$ \infer[\mathsf{one}]{\vdash \One}{} $

$\infer[\mathsf{cut}]{\vdash \Delta, \Delta'}{\infer{\vdash  \Delta, F}{\pi}& \infer{\vdash \Delta',F^\bot}{\pi'}} $

$\infer[\otimes]{\vdash \Delta, \Delta', F\otimes G}{\infer{\vdash  \Delta,F}{\pi}& \infer{\vdash \Delta',G}{\pi'}}  $

$\infer[\mathsf{bot}]{\vdash \Delta,\bot}{\infer{\vdash\Delta}{\pi}}$

$\infer[\w]{\vdash \Delta,\Xn}{\infer{\vdash\Delta}{\pi}}$

$\infer[\parr]{\vdash \Delta,F\parr G}{\infer{\vdash\Delta, F,G}{\pi}}$

$\infer[\cn]{\vdash \Delta,\Xn}{\infer{\vdash\Delta, \Xn , \Xn }{\pi}}$

$\infer[\mathsf{mix}]{\vdash \Delta, \Delta'}{\infer{\vdash  \Delta,}{\pi}& \infer{\vdash \Delta'\bot}{\pi'}} $

}

\section{Proofs of \Cref{sec:sem}}
\label{sec:proofs_sem}

	\begin{lemma}[\cref{lem:jointree}, Named paths] 
		Let $\pi : \Lambda$ be  a \pn  and {let $\X$  be the main name   of a box $\bbox^X$}. Then  in 
		$\pi$,  each edge with name    $X$ is  connected connected either to the  main conclusion of $\bbox^X$ or to a 
		negative conclusion $\Xn$   by a  directed path in which all  edges have  name $X$. There exists a unique such  path.
	\end{lemma}
	
	\begin{proof}Let us start with two easy observations. 

		\begin{enumerate}
			\item 
			By inspecting  the grammar of proof-net \Cref{fig:pn} and taking into account the polarized  orientation of the edges , we realize  that in
			$\pi$ each edge has \emph{at most one parent with the same underlying name}. Observe that 	the  only edges which do not have any parent with  the same name (resp., the same atom),  are either  the main edge of a \pax, or the negative conclusions of $\pi$. 
			
			\item From (1), 
			by recalling that a DAG where each node has
			at most one parent is a tree, we have the following key observation:  if we partition  ${\pi}$ into maximal 
			connected subgraphs whose vertices are \emph{edges with the same name $X$},
			then each such subgraph  is a \emph{directed tree}, whose root corresponds  either 
			to the main type of a box, or to a edge  in the context $\Lambda$ of the conclusion.

		\end{enumerate}

		Now let $X$ be the main name of a  \pax. Since $\pi$ is positive, there is no negative conclusion labelled by $\Xn$. Since we require
		the main names of $\pi$  to be pairwise distinct, we conclude that
		in ${\pi}$ there is \emph{exactly one} connected component of 
		edges with  the same name $X$.
	\end{proof}
	
	An immediate consequence is that if a name $X$ is summed out, then necessarily it is a main name.\todoc{??}

\begin{theorem}[\cref{th:compositional}, Graphs compositionality]
Let $\R=\Cuts{\R_1,\R_2}:\Delta$ be a \BPN. Then
\[\sem {\R} = \sum_{\bZ}  \left(  \sem {\R_1}  \odot \sem {\R_2}\right )
\quad
\text{ where } \bZ=  \big(\Nm{\sem{\R_1}} \cup \Nm{\sem{\N_1}}  \big) - \Nm{\Delta} \]
\end{theorem}
\begin{proof}
		Let us set $\Phi_{\M}\defeq\BigFProd_{\bbox\in \Boxes{\M}} \sem {\bbox}$ 
	for each  $\M$. We assume that every $\R_i$ has conclusions  $\Delta_i $.
	By  \Cref{def:sem}, we can write $\sem{\R_i}=\sum_{\bW_i} \Phi_{\R_i}$,
	where $\bW_i=\Nm{\Phi_{\R_i}} - \Nm{\Delta_i}$ ($i \in \{1,2\}$). 
	Crucially, by \Cref{cor:internal} , $\bW_1 \cap \Nm{{\R_2}} = \emptyset$ and $\bW_2 \cap \Nm{{\R_1}} = \emptyset$.
	Hence, by \Cref{eq:sum_prod}, sum and product distribute, and we have:
	\begin{align}\label{eq:binary_rule}
		\sem{\R_1} \Fprod \sem{\R_2} = & \sum_{\bW_1} \Phi_{\R_1} \FProd \sum_{\bW_2} \Phi_{\R_2}
		= \sum_{\bW_1} \sum_{\bW_2}  ( \Phi_{\R_1} \FProd  \Phi_{\R_2}) = \sum_{\bW_1} \sum_{\bW_2} \Phi_\R
	\end{align}
	where we used the fact that $\Boxes{\R} = \Boxes{\R_1} \cup \Boxes{\R_2}$.
	Now let $\bY_i\defeq\Nm{\sem{\R_i}}$; since $ \Nm{\Phi_{\R_i}} =  \bY_i \uplus \bW_i$ and,
	by \Cref{cor:internal}, $\bW_i \cap \Nm{\Phi_{\R_j}} =\emptyset$ for $i\not=j$, we have
	\begin{equation}\label{eq:names_prod}
		\Nm{\Phi_\R} = (\bY_1\cup \bY_2) \uplus (\bW_1\uplus \bW_2).
	\end{equation}
	Let $\bZ \defeq (\bY_1\cup\bY_2) - \Nm{\Delta}$; by \Cref{eq:names_prod} and \Cref{cor:internal}, 
	$\Nm{\Phi_\R} - \Nm{\Delta} = \bZ \uplus \bW_1 \uplus \bW_2 $. Therefore
	\[ \sum_\bZ	\big(\sem{\R_1} \Fprod \sem{\R_2} \big) = 
	\sum_\bZ \sum_{\bW_1} \sum_{\bW_2} \Phi_\R = \sum_{\Nm{\Phi_\R}-\Nm \Delta} \Phi_{\R} ~\defeq ~\sem{\R} \] 
	where we  sum out $\bZ$ from both sides of \Cref{eq:binary_rule}.
\end{proof}

\section{Proofs of Section \ref{sec:pruning}}\label{sec:proof_pruning}
\begin{lemma}[\cref{lem:weak_red_sem}]
Given a positive \BPN\ $\netR$, if $\netR\Red[\weak]\netS$ then $\sem{\netR}=\sem{\netS}$.
\end{lemma}
	\begin{proof}
		By definition, $\netR$ has a sub-net $\netN$ of the shape depicted on the left of \cref{fig:w_red}, and $\netS$ is the same graph with $\netN$ replaced by $\netM$ the shape depicted on the right of \cref{fig:w_red}.
		We now use \cref{def:decomposition}: we have $\pulledout{\netR}{\netN} = \Cuts{\netD,\netN}$ and $\pulledout{\netS}{\netM} = \Cuts{\netD,\netM}$ for some \pn $\netD$, as these two graphs differ only on $\netN$ and $\netM$.
		By compositionality (\cref{th:compositional}), we obtain:
		\begin{gather*}
			\sem{\netR} = \sem{\pulledout{\netR}{\netN}} = \sem{\netD} \Fprod \sem{\netN}
			\\
			\sem{\netS} = \sem{\pulledout{\netS}{\netM}} = \sem{\netD} \Fprod \sem{\netM}
		\end{gather*}
		
		We compute by \cref{def:sem} $\sem{\netM} = {\ftone}_{\emptyset}$ and $\sem{\netN} = \sum_X \phi$ with $\phi$ the \cpt associated to the box of $\netN$.
		But, $\phi$ being a \cpt on $X$, $\sum_X \phi = {\ftone}_{\{Y_1;\dots;Y_n\}}$.
		Thus:
		\begin{equation*}
			\sem{\netR} = \sem{\netD} \Fprod {\ftone}_{\{Y_1;\dots;Y_n\}} = \sem{\netD} = \sem{\netD} \Fprod {\ftone}_{\emptyset} = \sem{\netS}
		\end{equation*}
	\end{proof}

\section{Variable Elimination in \BPNs}
\label{sec:ve_in_bpn}

We sketch here an analog in \BPNs\ of the well-known \textit{variable elimination} algorithm~\cite{DarwicheBook}~(Chapter~6) used to compute efficiently the semantic of a \BN.
The idea the following.
Assume given a positive and atomic \BPN\ $\netC:\Gamma$ and an elimination order $(X_1, X_2, \dots, X_n)$ of all the \rvs in $\Nm{\netC}-\Nm{\Gamma}$.
We wish to compute in a specific way the semantic of $\netC$ which is $\sem{\netC} = \sum_{X_1, X_2, \dots, X_n} \left(\bigodot_{\bbox\in \Boxes{\netD}} \sem{\bbox} \right)$.
In this computation, we want to sum out first $X_1$ by distibuting the sum out on $X_1$ exactly on factors containing it.
This means we first compute $\sum_{X_1} \left(\bigodot_{\bbox\in \Boxes{\netC}\textit{ s.t. }X_1\in\Nm{\bbox}} \sem{\bbox} \right)$, which gives us a new factor.
Please remark this is less costly than computing the sum out of $X_1$ over all boxes.
We then iterate, summing out $X_2$ by again distributing its sum out only on factors containing it (so possibly on the new factor created by the summing out of $X_1$), and so on and so forth until summing out the last variable $X_n$.
As an example, look at \cref{fig:ABCDE3,fig:ABCDE4}: $\R_D'$ is a factorization of $\R_D$ with the order $(A,B,C)$.

To this aim, we expand $\netC$ into a cut-net $\netC'$ such that computing the interpretation of $\netC'$ using \cref{cor:turbo} does exactly that.
Such a $\netC'$ is obtained simply by repeated splittings (\cref{def:decomposition}) on well-chosen sub-nets.

\begin{lemma}
\label{prop:ve}
Take $\netC:\Gamma$ a positive and atomic \BPN\ and an elimination order $(X_1, \dots, X_n)$ of all the \rvs in $\Nm{\netC}-\Nm{\Gamma}$.
There is an expansion $\netC'$ of $\netC$ such that computing $\sem{\netC'}$ with \cref{cor:turbo} does exactly variable elimination for the given order.
\end{lemma}
\begin{proof}
The result follows by applying successively on $X_1$, \dots, $X_n$ the following claim, that corresponds to separating exactly the part of the net containing a given \rv\ $X$, while keeping track of all previous computations.
\begin{description}
\item[Claim]
Take a \rv\ $X$ and a positive and atomic
cut-net $\netD = \R(\M_1, \dots, \M_k, \N_1, \dots, \N_l)$ such that $X$ is in the conclusions
of all the $\M_i$ and none of the $\N_j$.
Then, there is a cut-net $\netD' = \R_1(\M', \N_1, \dots, \N_l)$ such that $\netD' \Reds[R] \netD$, $\M' = \R_2(\M_1, \dots, \M_k)$, $X \notin \Nm{\R_1}$ and for all $\boxn$-node $\bbox$ in $\R_2$, $X\in\Nm{\bbox}$.
\end{description}
Intuitively, the sub-nets $\M_1, \dots, \M_k, \N_1, \dots, \N_l$ correspond to already done computations, with the $\sem{\M_i}$ being those factors on (at least) $X$.
Meanwhile, the root $\R$ contains factors that we still have not used.
Our next computation is the summing out of $X$ in the product of all the factors $\sem{M_i}$ -- as well as the factors using $X$ among those of $\R$.

Thus, one can get the cut-net corresponding to the computation associated to variable elimination by starting from the full \BPN\ $\netC$ seen as a cut-net with a single sub-net (\ie\ all factors have not been used yet, everything is in the root), and apply the above result successively on $X_1$, then $X_2$, and so on until $X_n$.

We now prove our Claim.
Set $\R_2$ the minimal sub-net of $\R$ containing all its edges labelled by $X$: it is made of all $\cut$, $\ax$ and $\cn$ on $X$, and all boxes with input or output $X$.
Define $\M' \defeq \R_2(\M_1, \dots \M_k)$ (with the involved $\cut$-nodes of $\netD$ in-between).
We set $\D'$ as the cut-net $\D' \defeq \pulledout{\netD}{\M'} = \R_1(\M', \N_1, \dots, \N_l)$ whose root $\R_1$ is made of the nodes of $\R$ not in $\R_2$ (along the $\ax$-nodes added during the splitting).
By construction, we have all the required properties.
%
\end{proof}



%

\section{Proofs of \cref{sec:typing_cut_net}}
\label{sec:proof_typing}

\begin{theorem}[\cref{th:typingcutnet}]
Given a cut-net $\R = \CutS{\R_1, \dots, \R_n}$, there is a sequence of $\otimes/\parr$-expansions (the reverse of the $\otimes/\parr$-rule in \cref{fig:rewriting}) on the cuts between the pairs $(\R_i,\R_j)$ such that the result is a cut-net $\R' = \CutS{\R_1', \dots, \R_n'}$ with \textbf{at most one cut} between each pair $(\R_i', \R_j')$.
\end{theorem}
\begin{proof}
We will proceed by induction on a sequentialization $\pi$ of $\R$.
Call \textbf{proper} a cut-net with at most one $\cut$-node between each of its sub-nets.
We illustrate with a drawing each case, corresponding to what kind the last rule of $\pi$ is.
On these drawings, ellipses represent sub-graphs such as $\R_1$, \dots, $\R_n$ while bold edges between those represent $\cut$-nodes.
Up to commuting rules, we can assume $\pi$ has above a $mix$-rule only $\ax$-, $\One$-, $\boxn$- or $mix$-rules, so that the last rule of $\pi$ -- provided it is not already proper -- is either a $\parr$-, $\bot$-, $\cn$-, $\weak$-, $\otimes$- or $\cut$-rule.

\noindent$\bullet$
Assume $\pi$ is a proof $\rho$ followed by a $\parr$-rule, which corresponds in $\R$ to a node $v$.
Without loss of generality, $v$ is in $\R_1$.
Call $\N_1$ the removal of $v$ in $\R_1$.
By induction hypothesis on the cut-net $\N = \CutS{\N_1, \R_2, \dots, \R_n}$ of sequentialization $\rho$, there is a proper cut-net $\N' = \CutS{\N_1', \N_2', \dots, \N_n'}$ such that $\N' \Reds \N$.
Set $\R' \defeq \CutS{\N_1'\cup\{v\}, \N_2', \dots, \N_n'}$ with $\N_1'\cup\{v\}$ the graph obtained by addding $v$ to $\N_1'$ with the same edges it had in $\R_1$.
Then $\R'$ is a proper cut-net and $\R' \Reds \R$ using the same elimination steps as in $\N' \Reds \N$.

\begin{adjustbox}{}
\begin{tikzpicture}
	\node at (-5,2) {$\R$};
	\node at (0,2.5) {\color{red}$\N_1$};
	\draw[red, dashed] (0,2.5) ellipse (8mm and 4mm);
	\node at (-.75,1.5) {\color{lasallegreen}$\R_1$};
	\draw[lasallegreen, dashed] (0,1.9) ellipse (12mm and 11mm);
	\node at (2.5,2.5) {\color{lasallegreen}$\R_2$};
	\draw[lasallegreen, dashed] (2.5,2.5) ellipse (7mm and 6mm);
	\node at (-3,2.5) {\color{lasallegreen}$\R_3$};
	\draw[lasallegreen, dashed] (-3,2.5) ellipse (6mm and 5mm);
	\node at (4,3.5) {\color{lasallegreen}$\R_4$};
	\draw[lasallegreen, dashed] (4,3.5) ellipse (5mm and 6mm);
\begin{scope}[every node/.style={circle,minimum size=6mm,inner sep=0,draw=black,line width=0.8pt}, every edge/.style={draw=black,-,line width=0.8pt}]
	\node (v) at (0,1.5) {$\parr$};
	\path (-.5,2.25) edge node[named edge,left] {$A$} (v);
	\path (.5,2.25) edge node[named edge,right] {$B$} (v);
	\path (v) edge node[named edge] {$A\parr B$} (0,.75);
	\path (.6,2.65) edge (2,2.65);
	\path (.6,2.5) edge (2,2.5);
	\path (.6,2.35) edge (2,2.35);
	\path (-.5,2.65) edge (-2.75,2.5);
	\path (-.5,2.5) edge (-2.75,2.35);
	\path (2.5,3) edge (3.65,3.5);
	\path (2.65,2.9) edge (3.8,3.4);
\end{scope}

	\node at (0,0) {$\FromReds$};

	\node at (-5,-2) {$\R'$};
	\node at (0,-1.5) {\color{blue}$\N_1'$};
	\draw[blue, dashed] (0,-1.5) ellipse (8mm and 4mm);
	\node at (2.5,-1.5) {\color{blue}$\N_2'$};
	\draw[blue, dashed] (2.5,-1.5) ellipse (7mm and 6mm);
	\node at (-3,-1.5) {\color{blue}$\N_3'$};
	\draw[blue, dashed] (-3,-1.5) ellipse (6mm and 5mm);
	\node at (4,-.5) {\color{blue}$\N_4'$};
	\draw[blue, dashed] (4,-.5) ellipse (5mm and 6mm);
\begin{scope}[every node/.style={circle,minimum size=6mm,inner sep=0,draw=black,line width=0.8pt}, every edge/.style={draw=black,-,line width=0.8pt}]
	\node (v') at (0,-2.5) {$\parr$};
	\path (-.5,-1.75) edge node[named edge,left] {$A$} (v');
	\path (.5,-1.75) edge node[named edge,right] {$B$} (v');
	\path (v') edge node[named edge] {$A\parr B$} (0,-3.25);
	\path (.6,-1.5) edge (2,-1.5);
	\path (-.5,-1.5) edge (-2.75,-1.5);
	\path (2.5,-1) edge (3.65,-.5);
\end{scope}
\end{tikzpicture}
\end{adjustbox}

\noindent$\bullet$
The cases where the last rule of $\pi$ is a $\bot$-, $\cn$-, or $\weak$-rule are similar to the previous one.

\noindent$\bullet$
Suppose now the last rule of $\pi$ is a $\otimes$-rule between proofs $\rho^l$ and $\rho^r$, whose corresponding node in $\R$ we call $v$.
By symmetry, assume $v$ belongs to $R_1$.
Dividing $\R$ according to $\rho^l$ and $\rho^r$, we have $\R = \CutS{\N_1^l\cup\N_1^r\cup\{v\}, \N_2^l \cup \N_2^r, \dots, \N_n^l \cup \N_n^r}$. 
By induction hypothesis on the cut-net $\N^l \defeq \CutS{\N_1^l, \N_2^l, \dots, \N_n^l}$ of sequentialization $\rho^l$, there is a proper cut-net ${\N^l}' = \CutS{{\N_1^l}', {\N_2^l}', \dots, {\N_n^l}'}$ such that ${\N^l}' \Reds \N^l$; and similarly with $\rho^r$, one gets ${\N^r}' = \CutS{{\N_1^r}', {\N_2^r}', \dots, {\N_n^r}'} \Reds \N^r \defeq \CutS{\N_1^r, \N_2^r, \dots, \N_n^r}$.

The wished cut-net $\R'$ is obtained from ${\N^l}'$ and ${\N^r}'$ as follows.
Consider $\D \defeq \CutS{{\N_1^l}'\cup{\N_1^r}'\cup\{v\}, {\N_2^l}'\cup{\N_2^r}', \dots, {\N_n^l}'\cup{\N_n^r}'}$.
While $\D$ is a cut-net such that $\D \Reds \R$ (\eg\ by doing the steps in ${\N^l}' \Reds \N^l$ first then the ones in ${\N^r}' \Reds \N^r$), it may not be proper.
Indeed, there may be two $\cut$ between some ${\N_i^l}'\cup{\N_i^r}'$ and ${\N_j^l}'\cup{\N_j^r}'$, with one $\cut$ coming from ${\N^l}'$ and the other from ${\N^r}'$.
As the skeleton graph of $\D$ is a tree, say by symmetry that the component $i$ is closer to the first component -- the one containing $v$ -- than $j$.
For each such pair, we apply a $\parr-\otimes$ $\cut$-expansion, putting a $\parr$-node in the component $i$ and a $\otimes$-node in the component $j$.
The resulting components are correct since in $\D$ the only way to go from ${\N_j^l}'$ to ${\N_j^r}'$ passes through $v$, so through the two identified $\cut$.
Doing so for all possible $i$ and $j$ yields the wished $\R'$, with $\R' \Reds \D \Reds \R$.

\begin{adjustbox}{}
\begin{tikzpicture}
	\node at (-5,2) {$\R$};
	\node at (-.5,2.75) {\color{red}$\N_1^l$};
	\draw[red, dashed] (-.5,2.75) ellipse (8mm and 4mm);
	\node at (.75,2.2) {\color{red}$\N_1^r$};
	\draw[red, dashed] (.75,2.2) ellipse (7mm and 3mm);
	\node at (-1,1.5) {\color{lasallegreen}$\R_1$};
	\draw[lasallegreen, dashed] (-.25,1.9) ellipse (18mm and 13mm);
	\node at (4,2.9) {\color{red}$\N_2^l$};
	\draw[red, dashed] (4,2.9) ellipse (5mm and 3mm);
	\node at (4,2.1) {\color{red}$\N_2^r$};
	\draw[red, dashed] (4,2.1) ellipse (5mm and 3mm);
	\node at (4.6,2.5) {\color{lasallegreen}$\R_2$};
	\draw[lasallegreen, dashed] (4,2.5) ellipse (10mm and 7mm);
	\node at (-3.1,2.9) {\color{red}$\N_3^l$};
	\draw[red, dashed] (-3.1,2.9) ellipse (5mm and 3mm);
	\node at (-3.1,2.1) {\color{red}$\N_3^r$};
	\draw[red, dashed] (-3.1,2.1) ellipse (5mm and 3mm);
	\node at (-3.7,2.5) {\color{lasallegreen}$\R_3$};
	\draw[lasallegreen, dashed] (-3.1,2.5) ellipse (10mm and 7mm);
	\node at (7,3.9) {\color{red}$\N_4^l$};
	\draw[red, dashed] (7,3.9) ellipse (5mm and 3mm);
	\node at (7,3.1) {\color{red}$\N_4^r$};
	\draw[red, dashed] (7,3.1) ellipse (5mm and 3mm);
	\node at (7.6,3.5) {\color{lasallegreen}$\R_4$};
	\draw[lasallegreen, dashed] (7,3.5) ellipse (10mm and 7mm);
\begin{scope}[every node/.style={circle,minimum size=6mm,inner sep=0,draw=black,line width=0.8pt}, every edge/.style={draw=black,-,line width=0.8pt}]
	\node (v) at (0,1.5) {$\otimes$};
	\path (-.5,2.5) edge node[named edge,left] {$A$} (v);
	\path (.5,2.1) edge node[named edge,right] {$B$} (v);
	\path (v) edge node[named edge] {$A\otimes B$} (0,.75);
	\path (.1,2.85) edge (3.7,3);
	\path (.1,2.7) edge (3.7,2.85);
	\path (1.25,2.2) edge (3.7,2);
	\path (-1.1,2.85) edge (-2.75,3);
	\path (-1.1,2.7) edge (-2.75,2.85);
	\path (4.25,3) edge (6.8,3.8);
	\path (4.35,2) edge (7.1,2.9);
\end{scope}

	\node at (0,0) {$\FromReds$};

	\node at (-5,-2.4) {$\D$};
	\node at (-.5,-1.65) {\color{blue}${\N_1^l}'$};
	\draw[blue, dashed] (-.5,-1.65) ellipse (8mm and 4mm);
	\node at (.75,-2.2) {\color{blue}${\N_1^r}'$};
	\draw[blue, dashed] (.75,-2.2) ellipse (7mm and 3mm);
	\node at (4,-1.5) {\color{blue}${\N_2^l}'$};
	\draw[blue, dashed] (4,-1.5) ellipse (5mm and 3mm);
	\node at (4,-2.3) {\color{blue}${\N_2^r}'$};
	\draw[blue, dashed] (4,-2.3) ellipse (5mm and 3mm);
	\node at (-3.1,-1.5) {\color{blue}${\N_3^l}'$};
	\draw[blue, dashed] (-3.1,-1.5) ellipse (5mm and 3mm);
	\node at (-3.1,-2.3) {\color{blue}${\N_3^r}'$};
	\draw[blue, dashed] (-3.1,-2.3) ellipse (5mm and 3mm);
	\node at (7,-.5) {\color{blue}${\N_4^l}'$};
	\draw[blue, dashed] (7,-.5) ellipse (5mm and 3mm);
	\node at (7,-1.3) {\color{blue}${\N_4^r}'$};
	\draw[blue, dashed] (7,-1.3) ellipse (5mm and 3mm);
\begin{scope}[every node/.style={circle,minimum size=6mm,inner sep=0,draw=black,line width=0.8pt}, every edge/.style={draw=black,-,line width=0.8pt}]
	\node (v') at (0,-2.9) {$\otimes$};
	\path (-.5,-1.9) edge node[named edge,left] {$A$} (v');
	\path (.5,-2.3) edge node[named edge,right] {$B$} (v');
	\path (v') edge node[named edge] {$A\otimes B$} (0,-3.65);
	\coordinate (c13l) at (.1,-1.7);
	\coordinate (c13r) at (1.25,-2.2);
	\coordinate (c31l) at (3.7,-1.4);
	\coordinate (c31r) at (3.7,-2.4);
	\path (c13l) edge (c31l);
	\path (c13r) edge (c31r);
	\path (-1.1,-1.7) edge (-2.75,-1.4);
	\coordinate (c34l) at (4.25,-1.4);
	\coordinate (c34r) at (4.35,-2.4);
	\coordinate (c43l) at (6.8,-.6);
	\coordinate (c43r) at (7,-1.5);
	\path (c34l) edge (c43l);
	\path (c34r) edge (c43r);
\end{scope}

	\node at (0,-4) {$\FromReds$};

	\node at (-5,-5.4) {$\R'$};
	\node at (-.5,-4.65) {\color{blue}${\N_1^l}'$};
	\draw[blue, dashed] (-.5,-4.65) ellipse (8mm and 4mm);
	\node at (.75,-5.2) {\color{blue}${\N_1^r}'$};
	\draw[blue, dashed] (.75,-5.2) ellipse (7mm and 3mm);
	\node at (4,-4.5) {\color{blue}${\N_2^l}'$};
	\draw[blue, dashed] (4,-4.5) ellipse (5mm and 3mm);
	\node at (4,-5.3) {\color{blue}${\N_2^r}'$};
	\draw[blue, dashed] (4,-5.3) ellipse (5mm and 3mm);
	\node at (-3.1,-4.5) {\color{blue}${\N_3^l}'$};
	\draw[blue, dashed] (-3.1,-4.5) ellipse (5mm and 3mm);
	\node at (-3.1,-5.3) {\color{blue}${\N_3^r}'$};
	\draw[blue, dashed] (-3.1,-5.3) ellipse (5mm and 3mm);
	\node at (7,-3.5) {\color{blue}${\N_4^l}'$};
	\draw[blue, dashed] (7,-3.5) ellipse (5mm and 3mm);
	\node at (7,-4.3) {\color{blue}${\N_4^r}'$};
	\draw[blue, dashed] (7,-4.3) ellipse (5mm and 3mm);
\begin{scope}[every node/.style={circle,minimum size=6mm,inner sep=0,draw=black,line width=0.8pt}, every edge/.style={draw=black,-,line width=0.8pt}]
	\node (v') at (0,-5.9) {$\otimes$};
	\path (-.5,-4.9) edge node[named edge,left] {$A$} (v');
	\path (.5,-5.3) edge node[named edge,right] {$B$} (v');
	\path (v') edge node[named edge] {$A\otimes B$} (0,-6.65);
	\coordinate (c13l) at (.1,-4.7);
	\coordinate (c13r) at (1.25,-5.2);
	\coordinate (c31l) at (3.7,-4.4);
	\coordinate (c31r) at (3.7,-5.4);
	\node[draw=red] (t13) at (2,-4.9) {\color{red}$\parr$};
	\node[draw=red] (t31) at (3,-4.9) {\color{red}$\otimes$};
	\path (c13l) edge[draw=red] (t13);
	\path (c13r) edge[draw=red] (t13);
	\path (t13) edge[draw=red] (t31);
	\path (c31l) edge[draw=red] (t31);
	\path (c31r) edge[draw=red] (t31);
	\path (-1.1,-4.7) edge (-2.75,-4.4);
	\coordinate (c34l) at (4.25,-4.4);
	\coordinate (c34r) at (4.35,-5.4);
	\coordinate (c43l) at (6.8,-3.6);
	\coordinate (c43r) at (7,-4.5);
	\node[draw=red] (t34) at (5,-4.7) {\color{red}$\parr$};
	\node[draw=red] (t43) at (5.9,-4.4) {\color{red}$\otimes$};
	\path (c34l) edge[draw=red] (t34);
	\path (c34r) edge[draw=red] (t34);
	\path (t34) edge[draw=red] (t43);
	\path (c43l) edge[draw=red] (t43);
	\path (c43r) edge[draw=red] (t43);
\end{scope}
\end{tikzpicture}
\end{adjustbox}

\noindent$\bullet$
The case where the last rule of $\pi$ is a $\cut$-rule \emph{whose corresponding node is not a $\cut$ between some $\R_i$ and $\R_j$} is similar to the $\otimes$ case.

\noindent$\bullet$
The last case is the one where the last rule of $\pi$ is a $\cut$-rule between proofs $\rho^l$ and $\rho^r$, with the node $v$ corresponding to this $\cut$ being, by symmetry, between $\R_1$ and $\R_2$.
Our reasonning is similar to the $\otimes$ case, with an additional step at the end.
Dividing $\R$ according to $\rho^l$ and $\rho^r$, we have $\R = \CutS{\N_1^l\cup\N_1^r, \N_2^l \cup \N_2^r, \dots, \N_n^l \cup \N_n^r}$.
In particular, $v$ is a $\cut$ with a premise in $\N_1^l$ and the other in $\N_2^r$.
By induction hypothesis on the cut-net $\N^l \defeq \CutS{\N_1^l, \N_2^l, \dots, \N_n^l}$ of sequentialization $\rho^l$, there is a proper cut-net ${\N^l}' = \CutS{{\N_1^l}', {\N_2^l}', \dots, {\N_n^l}'}$ such that ${\N^l}' \Reds \N^l$; and similarly one gets ${\N^r}' = \CutS{{\N_1^r}', {\N_2^r}', \dots, {\N_n^r}'} \Reds \N^r \defeq \CutS{\N_1^r, \N_2^r, \dots, \N_n^r}$.

The wished cut-net $\R'$ is obtained from ${\N^l}'$ and ${\N^r}'$ as follows.
Consider $\D \defeq \CutS{{\N_1^l}'\cup{\N_1^r}', {\N_2^l}'\cup{\N_2^r}', \dots, {\N_n^l}'\cup{\N_n^r}'}$.
This cut-net indeed reduces to $\R$, but it may not be proper: not only may there be two $\cut$ between some ${\N_i^l}'\cup{\N_i^r}'$ and ${\N_j^l}'\cup{\N_j^r}'$ with one $\cut$ coming from ${\N^l}'$ and the other from ${\N^r}'$, but the $\cut$-node $v$ may introduce a second or third $\cut$ between ${\N_1^l}'\cup{\N_1^r}'$ and ${\N_2^l}'\cup{\N_2^r}'$.
The first case, which may happen for a pair $(i,j) \neq (1,2)$, is solved exactly as in the $\otimes$ case above.
Thus, let us focus on the case where there already are $\cut$ between the first two components in ${\N^l}'$ and ${\N^r}'$.
We have two sub-cases: either only one of ${\N^l}'$ and ${\N^r}'$, say ${\N^l}'$, has such a $\cut$; or both have one.

If only ${\N^l}'$ has a $\cut$-node between ${\N_1^l}'$ and ${\N_2^l}'$, then the only way to go from ${\N_2^l}'$ to ${\N_2^r}'$ passes through $v$.
We can thus apply a $\parr-\otimes$ $\cut$-expansion, putting a $\parr$-node in ${\N_1^l}'$ and a $\otimes$-node in ${\N_2^l}' \cup {\N_2^r}'$, which yields the wished $\R'$.

Now, suppose there are $\cut$-nodes $u_l$ between ${\N_1^l}'$ and ${\N_2^l}'$ and a $u_r$ between ${\N_1^r}'$ and ${\N_2^r}'$.
In this case, we do the two following $\parr-\otimes$ $\cut$-expansion steps.
First, we expand the $\cut$ $v$ with $u_r$ by putting a $\otimes$-node in ${\N_1^l}' \cup {\N_1^r}'$ and a $\parr$-node in ${\N_2^r}'$, with a new $\cut$-node $v'$. The result is correct as paths from ${\N_1^l}'$ to ${\N_1^r}'$ must go through $v$.
Then, we expand $v'$ with $u_l$ by putting a $\parr$-node in ${\N_1^l}' \cup {\N_1^r}' \cup \{\otimes\}$ and a $\otimes$-node in ${\N_2^l}' \cup ({\N_2^r}' \cup \{\parr\})$. Again, the result is correct as paths from ${\N_2^l}'$ to ${\N_2^r}' \cup \{\parr\}$ must go through $v'$.
This gives us a proper cut-net $\R'$ such that $\R' \Reds \D$, and concludes the proof.\footnote{The apparent asymmetry between $u_l$ and $u_r$ is not relevant: it is also possible to expand $v$ first with $u_l$ and then with $u_r$, getting another proper cut-net reducing to $\R$.}

\begin{adjustbox}{}
\begin{tikzpicture}
	\node at (-5,2.5) {$\R$};
	\node at (-.5,2.9) {\color{red}$\N_1^l$};
	\draw[red, dashed] (-.5,2.9) ellipse (5mm and 3mm);
	\node at (-.5,2.1) {\color{red}$\N_1^r$};
	\draw[red, dashed] (-.5,2.1) ellipse (5mm and 3mm);
	\node at (-1.1,2.5) {\color{lasallegreen}$\R_1$};
	\draw[lasallegreen, dashed] (-.5,2.5) ellipse (10mm and 7mm);
	\node at (4,2.9) {\color{red}$\N_2^l$};
	\draw[red, dashed] (4,2.9) ellipse (5mm and 3mm);
	\node at (4,2.1) {\color{red}$\N_2^r$};
	\draw[red, dashed] (4,2.1) ellipse (5mm and 3mm);
	\node at (4.6,2.5) {\color{lasallegreen}$\R_2$};
	\draw[lasallegreen, dashed] (4,2.5) ellipse (10mm and 7mm);
	\node at (-3.1,2.9) {\color{red}$\N_3^l$};
	\draw[red, dashed] (-3.1,2.9) ellipse (5mm and 3mm);
	\node at (-3.1,2.1) {\color{red}$\N_3^r$};
	\draw[red, dashed] (-3.1,2.1) ellipse (5mm and 3mm);
	\node at (-3.7,2.5) {\color{lasallegreen}$\R_3$};
	\draw[lasallegreen, dashed] (-3.1,2.5) ellipse (10mm and 7mm);
	\node at (7,3.9) {\color{red}$\N_4^l$};
	\draw[red, dashed] (7,3.9) ellipse (5mm and 3mm);
	\node at (7,3.1) {\color{red}$\N_4^r$};
	\draw[red, dashed] (7,3.1) ellipse (5mm and 3mm);
	\node at (7.6,3.5) {\color{lasallegreen}$\R_4$};
	\draw[lasallegreen, dashed] (7,3.5) ellipse (10mm and 7mm);
\begin{scope}[every node/.style={circle,minimum size=6mm,inner sep=0,draw=black,line width=0.8pt}, every edge/.style={draw=black,-,line width=0.8pt}]
	\node (v) at (1.8,2.5) {$\cut$};
	\path (-.2,2.9) edge (v);
	\path (3.7,2.1) edge (v);
	\path (-.1,3.05) edge (3.7,3);
	\path (-.1,2.2) edge (3.7,2);
	\path (-.75,2.9) edge (-2.75,3);
	\path (-.75,2.75) edge (-2.75,2.85);
	\path (4.25,3) edge (6.8,3.8);
	\path (4.35,2) edge (7.1,2.9);
\end{scope}

	\node at (1.8,0) {$\FromReds$};

	\node at (-5,-1.9) {$\D$};
	\node at (-.5,-1.5) {\color{blue}${\N_1^l}'$};
	\draw[blue, dashed] (-.5,-1.5) ellipse (5mm and 3mm);
	\node at (-.5,-2.3) {\color{blue}${\N_1^r}'$};
	\draw[blue, dashed] (-.5,-2.3) ellipse (5mm and 3mm);
	\node at (4,-1.5) {\color{blue}${\N_2^l}'$};
	\draw[blue, dashed] (4,-1.5) ellipse (5mm and 3mm);
	\node at (4,-2.3) {\color{blue}${\N_2^r}'$};
	\draw[blue, dashed] (4,-2.3) ellipse (5mm and 3mm);
	\node at (-3.1,-1.5) {\color{blue}${\N_3^l}'$};
	\draw[blue, dashed] (-3.1,-1.5) ellipse (5mm and 3mm);
	\node at (-3.1,-2.3) {\color{blue}${\N_3^r}'$};
	\draw[blue, dashed] (-3.1,-2.3) ellipse (5mm and 3mm);
	\node at (7,-.5) {\color{blue}${\N_4^l}'$};
	\draw[blue, dashed] (7,-.5) ellipse (5mm and 3mm);
	\node at (7,-1.3) {\color{blue}${\N_4^r}'$};
	\draw[blue, dashed] (7,-1.3) ellipse (5mm and 3mm);
\begin{scope}[every node/.style={circle,minimum size=6mm,inner sep=0,draw=black,line width=0.8pt}, every edge/.style={draw=black,-,line width=0.8pt}]
	\coordinate (c12l) at (-.1,-1.4);
	\coordinate (c12r) at (-.1,-2.4);
	\coordinate (c21l) at (3.7,-1.4);
	\coordinate (c21r) at (3.7,-2.4);
	\coordinate (c1v) at (-.1,-1.6);
	\coordinate (c2v) at (3.7,-2.2);
	\path (c12l) edge (c21l);
	\path (c12r) edge (c21r);
	\node (v') at (1.8,-1.9) {$\cut$};
	\path (c1v) edge (v');
	\path (c2v) edge (v');
	\path (-.9,-1.4) edge (-2.75,-1.4);
	\coordinate (c34l) at (4.25,-1.4);
	\coordinate (c34r) at (4.35,-2.4);
	\coordinate (c43l) at (6.8,-.6);
	\coordinate (c43r) at (7,-1.5);
	\path (c34l) edge (c43l);
	\path (c34r) edge (c43r);
\end{scope}

	\node at (1.8,-3.5) {$\FromReds$};

	\node at (-5,-4.9) {$\R'$};
	\node at (-.5,-4.5) {\color{blue}${\N_1^l}'$};
	\draw[blue, dashed] (-.5,-4.5) ellipse (5mm and 3mm);
	\node at (-.5,-5.3) {\color{blue}${\N_1^r}'$};
	\draw[blue, dashed] (-.5,-5.3) ellipse (5mm and 3mm);
	\node at (4,-4.5) {\color{blue}${\N_2^l}'$};
	\draw[blue, dashed] (4,-4.5) ellipse (5mm and 3mm);
	\node at (4,-5.3) {\color{blue}${\N_2^r}'$};
	\draw[blue, dashed] (4,-5.3) ellipse (5mm and 3mm);
	\node at (-3.1,-4.5) {\color{blue}${\N_3^l}'$};
	\draw[blue, dashed] (-3.1,-4.5) ellipse (5mm and 3mm);
	\node at (-3.1,-5.3) {\color{blue}${\N_3^r}'$};
	\draw[blue, dashed] (-3.1,-5.3) ellipse (5mm and 3mm);
	\node at (7,-3.5) {\color{blue}${\N_4^l}'$};
	\draw[blue, dashed] (7,-3.5) ellipse (5mm and 3mm);
	\node at (7,-4.3) {\color{blue}${\N_4^r}'$};
	\draw[blue, dashed] (7,-4.3) ellipse (5mm and 3mm);
\begin{scope}[every node/.style={circle,minimum size=6mm,inner sep=0,draw=black,line width=0.8pt}, every edge/.style={draw=black,-,line width=0.8pt}]
	\coordinate (c12l) at (-.1,-4.4);
	\coordinate (c12r) at (-.1,-5.4);
	\coordinate (c21l) at (3.7,-4.4);
	\coordinate (c21r) at (3.7,-5.4);
	\coordinate (c1v) at (-.1,-4.6);
	\coordinate (c2v) at (3.7,-5.2);
	\node[draw=red] (t1) at (.7,-5.2) {\color{red}$\otimes$};
	\node[draw=red] (p1) at (1.3,-4.5) {\color{red}$\parr$};
	\node[draw=red] (p2) at (2.9,-5.2) {\color{red}$\parr$};
	\node[draw=red] (t2) at (2.3,-4.5) {\color{red}$\otimes$};
	\path (c12r) edge[draw=red] (t1);
	\path (c1v) edge[draw=red] (t1);
	\path (c12l) edge[draw=red] (p1);
	\path (t1) edge[draw=red] (p1);
	\path (c21r) edge[draw=red] (p2);
	\path (c2v) edge[draw=red] (p2);
	\path (c21l) edge[draw=red] (t2);
	\path (p2) edge[draw=red] (t2);
	\path (p1) edge[draw=red] (t2);
	\path (-.9,-4.4) edge (-2.75,-4.4);
	\coordinate (c34l) at (4.25,-4.4);
	\coordinate (c34r) at (4.35,-5.4);
	\coordinate (c43l) at (6.8,-3.6);
	\coordinate (c43r) at (7,-4.5);
	\node[draw=red] (t34) at (5,-4.7) {\color{red}$\parr$};
	\node[draw=red] (t43) at (5.9,-4.4) {\color{red}$\otimes$};
	\path (c34l) edge[draw=red] (t34);
	\path (c34r) edge[draw=red] (t34);
	\path (t34) edge[draw=red] (t43);
	\path (c43l) edge[draw=red] (t43);
	\path (c43r) edge[draw=red] (t43);
\end{scope}
\end{tikzpicture}
\end{adjustbox}
\end{proof}

\newpage


\newcommand{\art}[1]{\widehat{#1}}
\newcommand{\tpnormal}{$\otimes/\parr$-normal\xspace}

\section{BPN artifact criterion}

By \pn we always mean a \pn with probabilistic boxes which is  \emph{well-named}, \ie	\emph{the atoms labelling the positive conclusions of the boxes are \emph{pairwise distinct}}. \todoc{ TODO: change also  in the body}


We  recall \Cref{lem:jointree}

	%

\begin{lemma}[Jointree-like property] 
	Let $\R : \Lambda$ be a \emph{positive} atomic   \pn. The restriction of $\R$ to the  edges labelled by a  same atom  $X$ is a   \emph{directed} tree, where
	\begin{enumerate}
		\item the  root is  the main conclusion of the box  $\bbox^X$;
		\item the leaves are either positive conclusion of $\R$, or inputs of a box.
	\end{enumerate}
	
\end{lemma}

\subsection{Properties of directed paths}
We will use some technical properties of directed paths that we collect below.

By Polarized Correctness (\Cref{lem:pol_correct}), the following holds \todocf{A direct proof is easy to give, and probably is better to do so}
\begin{proposition}[directed cycles]\label{lem:cycles} Let $\R$  be an atomic  
  net.
  $\R$ has a switching cycle iff $\R$  has a directed cycle (following the orientation of the edges).
\end{proposition}

	\begin{lemma}[directed paths]\label{lem:paths}  Let $\R$  be an atomic net (not necessarily correct). 
		\begin{itemize}
			\item An $X$-path with source  $\bbox^X$ cannot enter a cycle.
			\item Given an $X$-path, only its source and target nodes may  be boxes.
			\item   Any two consecutive edges $\xrightarrow{e_1} \xrightarrow{e_2}$ have different names iff the common vertex is a box. In this case, $e_2$ is the main conclusion  of the box. 
		\end{itemize}
	\end{lemma}
	
	As a consequence: 
	\begin{corollary}[boxes cycle]\label{lem:boxes_cycle} Let  $\R$  be an atomic  
		net with a directed cycle $r$. If $r$ visits at least one  box, then $r$	necessarily consists in a concatenation of $k>1$ directed  paths $r^{X_1} r^{X_2} ... r^{X_k}$  where each  $r^{X_i}$  is an $X_i$-path with source $\bbox^{X_i}$ and target $\bbox^{X_{i+1}}$ (where $X_{k+1} =X_1$).  
	\end{corollary}

	
	\subsection{Artifact criterion for  atomic \BPNs}

	\begin{definition}[Bayesian proof-net (\BPN)]
	A \pn $\R:\Delta$    is 	 Bayesian if 
		\begin{enumerate}
			\item  all the atoms in the conclusion $\Delta$ are positive, or
			\item $\R$ is sub-net of a \pn which satisfies (1).
		\end{enumerate}
	\end{definition}

%
%
	
	\begin{remark}\label{rem:neg1}
	Observe that if no negative conclusions of $\R$ has a name which is also the main name of a box,  then $\R$  is  a \BPN. Indeed, we can always obtain a positive \pn of which  $\R$ is a sub-net. If all conclusions of $\R$ are atomic, we  perform the following: 
	\begin{itemize}
	
		\item  	collect  together all the negative conclusions $e_1:\Xn, \dots, e_k:\Xn$ labelled by the same $\Xn$ as premises of a contraction of single conclusion   $e:\Xn$.
		\item  add a box $\bbox^X$ of single conclusion $f:\X$,  and a cut between $f:\X$ and $e:\Xn$.
	\end{itemize}
If $\R$ has any non-atomic conclusion $F$, we first cut every such  $F$ with the expansion of $F^\bot$ , then proceed as above.
	\end{remark}

\todocf{It may be worth (or not) to define the atomization of a conclusion}
	
	\begin{definition}[Artifact]
	Given an  atomic  \pn   $\R$ of  conclusions $\Delta$,  its artifact closure $\art\R$ is defined by performing the following for  
	each name $X$ such that  $X^-\in {\Delta}$    and  $\bbox^X$ belongs to $\R$:
	
\begin{itemize}
	\item 	contract together all the negative conclusions $e_1:\Xn, \dots, e_k:\Xn$ labelled by $\Xn$.  Let  $e:\Xn$ be the    conclusion of the contraction.
	\item if there exists (at least) one  positive conclusion $f:\X\in \Delta$ which is connected to $\bbox^X$ by  an $X$-path, 
	add a cut between  $e:\Xn$ and $f:\X$. Otherwise,  $\art\R$ is not defined. 

\end{itemize}
$\R$	satisfies the artifact test if  $\art\R$ is defined, and it is a \pn.
	\end{definition}

\begin{proposition} An atomic \pn $\R$   satisfies the artifact test \textit{ if and only if } $\R$ is a \BPN. 
\end{proposition}	
	
	\begin{proof}
\textit{Only if.} By hypothesis,  $\art{\R}$ is a \pn. For all its negative conclusions, we proceed as in \Cref{rem:neg1}, obtaining a positive \pn which includes $\R$.

\textit{If.} Assume that there exists a positive \pn $\netN$ of which $\R$ is a subnet. We prove that if $\art{\R}$ has a switching  cycle, so has $\netN$, contradicting the fact that it is a \pn.

By \Cref{lem:cycles},  $\art{\R}$ contains a  directed cycle $r$, which (since $\R$ is a \pn) necessarily uses one of the artifact cuts. Let $f:\X$ be the positive edge  of such a cut. By construction, there is an $X$-path from $\bbox^X$ to $f:\X$. Hence,  by  \Cref{lem:boxes_cycle}, the directed cycle visits $k>1$ boxes. Assume that the order of visit is $\mathcal C= \bbox^{X_1} \rightarrow \dots, \rightarrow \bbox^{X_k} \rightarrow\bbox^{X_1}$. Observe that each $\bbox$ in the cycle has a negative edge with the same name 
as the previously visited box. 

 In $\netN$, we have the same boxes.  Observe that  for every  pair in $\mathcal C$  of consecutive boxes  $\bbox^X \rightarrow \bbox^Y$, by the jointree property (\Cref{lem:jointree}) in $\netN$ there is an $X$-path connecting   the negative conclusion $e:\Xn$ of $\bbox^Y$ with   $\bbox^X$. The concatenation of all such paths gives a directed  cycle in $\netN$.
	\end{proof}

%
%




\end{document}